\documentclass[conference]{new-aiaa}
\usepackage[utf8]{inputenc}

\usepackage[fancythm,fancybb]{jphmacros2e} 
\usepackage{subcaption}
\usepackage{graphicx}
\usepackage{tikz}
\usepackage{tikz-3dplot}
\usepackage{amsmath}
\usepackage[version=4]{mhchem}
\usepackage{siunitx}
\usepackage{longtable,tabularx}
\newcommand{\norm}[1]{\left\lVert#1\right\rVert}
\usepackage{comment}
\usepackage{footmisc}
\usepackage{hyperref}
\usepackage{algorithm}
\usepackage{algpseudocode}
\usepackage{mathtools}
\usepackage{cleveref}[2012/02/15]
\allowdisplaybreaks
\usepackage{bm}  % in preamble

\usepackage{lettrine}

\makeatletter
\newcommand\footnoteref[1]{\protected@xdef\@thefnmark{\ref{#1}}\@footnotemark}
\makeatother

\makeatletter
\@ifundefined{thefootnote}{%
  \newcounter{footnote}%
  \renewcommand\thefootnote{\arabic{footnote}}%
}{}
\makeatother

\newcommand{\mf}[1]{\mathbf{#1}}
\newcommand{\mr}[1]{\mathrm{#1}}

\title{Adversarial Pursuits in Cislunar Space}

\author{Filippos Fotiadis$^{\star}$\footnote{Postdoctoral Researcher, Oden Institute for Computational Engineering \& Sciences, e-mail: ffotiadis@utexas.edu}, Quentin Rommel$^{\star}$\footnote{Graduate Research Assistant, Department of Aerospace Engineering \& Engineering Mechanics, e-mail: quentin.rommel@utmail.utexas.edu}}
\affil{The University of Texas at Austin, Austin, TX, 78712}
\author{Gregory Falco\footnote{Assistant Professor, Department of Mechanical and Aerospace Engineering, e-mail: gfalco@cornell.edu}}
\affil{Cornell University, Ithaca, NY, 14850}
\author{Ufuk Topcu\footnote{Professor, Department of Aerospace Engineering \& Engineering Mechanics, e-mail: utopcu@utexas.edu}}
\affil{The University of Texas at Austin, Austin, TX, 78712}
\begin{document}

\maketitle

\makeatletter
\let\svthefootnote\thefootnote
\let\thefootnote\relax
\footnotetext{$^{\star}$ Equal contribution}
\let\thefootnote\svthefootnote
\makeatother

\begin{abstract}
Cislunar space is becoming a critical domain for future lunar and interplanetary missions, yet its remoteness, sparse infrastructure, and unstable dynamics create single points of failure. Adversaries in cislunar orbits can exploit these vulnerabilities to pursue and jam co-located communication relays, potentially severing communications between lunar missions and the Earth. We study a pursuit-evasion scenario between two spacecraft in a cislunar orbit, where the evader must avoid a pursuer-jammer while remaining close to its nominal trajectory. We model the evader-pursuer interaction as a zero-sum adversarial differential game cast in the circular restricted three-body problem. This formulation incorporates critical aspects of cislunar orbital dynamics, including autonomous adjustment of the reference orbit phasing to enable aggressive evading maneuvers, and shaping of the evader’s cost with the orbit’s stable and unstable manifolds. We solve the resulting nonlinear game locally using a continuous-time differential dynamic programming variant, which iteratively applies linear-quadratic approximations to the Hamilton-Jacobi-Isaacs equation. We simulate the evader’s behavior against both a worst-case and a linear-quadratic pursuer. Our results pave the way for securing future missions in cislunar space against emerging cyber threats.
\end{abstract}

\section{Nomenclature}

{\renewcommand\arraystretch{1.0}
\noindent\begin{longtable*}{@{}l @{\quad=\quad} l@{}}
CR3BP &    Circular Restricted Three-Body Problem \\
DDP  & Differential Dynamic Programming \\
HJI & Hamilton-Jacobi-Isaacs \\
IMU & Inertial Measurement Unit \\
$L_1/L_2$ & Earth-Moon collinear Lagrange points\\
NASA & National Aeronautics and Space Administration
\end{longtable*}}

\section{Introduction}

\lettrine[lines=2,lhang=0.1,nindent=0em]{C}islunar space is becoming an important region for upcoming lunar and
 interplanetary missions \cite{crusan2019nasa, baker2024comprehensive, holzinger2021primer}, yet its remote and chaotic nature creates unique cybersecurity challenges. Unlike traditional low-Earth or geosynchronous orbits, where redundancy is easier to afford, cislunar missions will be comparatively sparse and heavily reliant on long-distance communication links \cite{badura2023optimizing, kurtsecurity}, creating potential single points of failure. This gives adversaries novel opportunities to pursue target spacecraft, disrupt radio transmissions, and, by exploiting the unstable dynamics of cislunar space, drive them off their nominal orbits. Ensuring a safe, long-term presence in cislunar space thus requires addressing these orbital threats, particularly through evasive strategies that allow spacecraft to maneuver away from hostile assets.

Recent in-orbit offensive maneuvers highlight the urgent need for defensive strategies against pursuits. A well-known example occurred in 2015, when Russia’s “Luch/Olymp-K” satellite maneuvered unusually close to commercial communications satellites in geostationary orbit \cite{roberts2024method, sankaran2022russia, roberts2020sustainable}. To deal with such pursuits, pursuit-evasion games have provided a mathematical method for modeling adversarial encounters, framing the interaction as a zero-sum game in which the pursuer seeks interception while the evader tries to escape \cite{bacsar1998dynamic, isaacs1999differential}.  Such games have been successfully applied in the context of Earth orbits and under the assumption of Keplerian motion \cite{fu2025analytical, shen2018revisit, mehlman2024cat}. Nevertheless, in cislunar space, the dynamics are typically modeled by the circular restricted three-body problem (CR3BP), which is nonlinear and unstable \cite{koon2000dynamical, connor1984three, szebehely1967}. This makes classical two-body Keplerian approaches, particularly those based on local frames, inapplicable.

Lagrange points in cislunar space provide key benefits for navigation and communication, yet the complex dynamics of the CR3BP require frequent stationkeeping to stay in their vicinity \cite{gomez1998station}. This can make evasive maneuvers against pursuers challenging and risky to perform. On the other hand, cyber-physical attacks like jamming or spoofing also become difficult to carry out as they are easier to detect and require precise pointing. However, the payoff potential is extraordinarily high due to the limited defense options. This risk and reward balance makes it critical to develop pursuit-evasion strategies for spacecraft operating near Lagrange-point orbits.

A spacecraft in cislunar space must guard against three broad types of pursuits. One of the simplest is proximity‐based interference as presented in Figure \ref{fig:cislunar_jamming_diagram}: an attacker drifts close enough to jam communications or interfere with onboard sensors. Because ground‐station links are already weak at lunar distances, even low‐power interference can prevent precise communication-based orbit determination, resulting in a higher risk for the spacecraft to drift away from its mission. A second, more aggressive threat is direct interception, or the “kamikaze” attack. Here, the attacker maneuvers to match the spacecraft’s position (and often its velocity) precisely, with the intent to collide or destroy.  Between these extremes lies the third, rendezvous‐and‐inspection threat. By aligning both position and velocity with the target spacecraft, an adversary can hover alongside to observe or subtly nudge the target off its nominal orbit before moving on to new objectives. Small thrusts or brief sensor interference during such proximity operations can introduce navigation errors that grow over time. Because satellites in $L_1$/$L_2$ orbits rely heavily on ground‐based radiometric tracking to correct onboard sensor drift, any jamming or spoofing will allow chaotic dynamics to amplify even tiny errors into mission‐ending deviations. Cross‐checks among optical navigation, IMUs, and other onboard measurements may not be precise enough to compensate once external interference begins, and the resulting errors will become harder to correct over time.
\begin{figure}
    \centering
    \includegraphics[width=\linewidth]{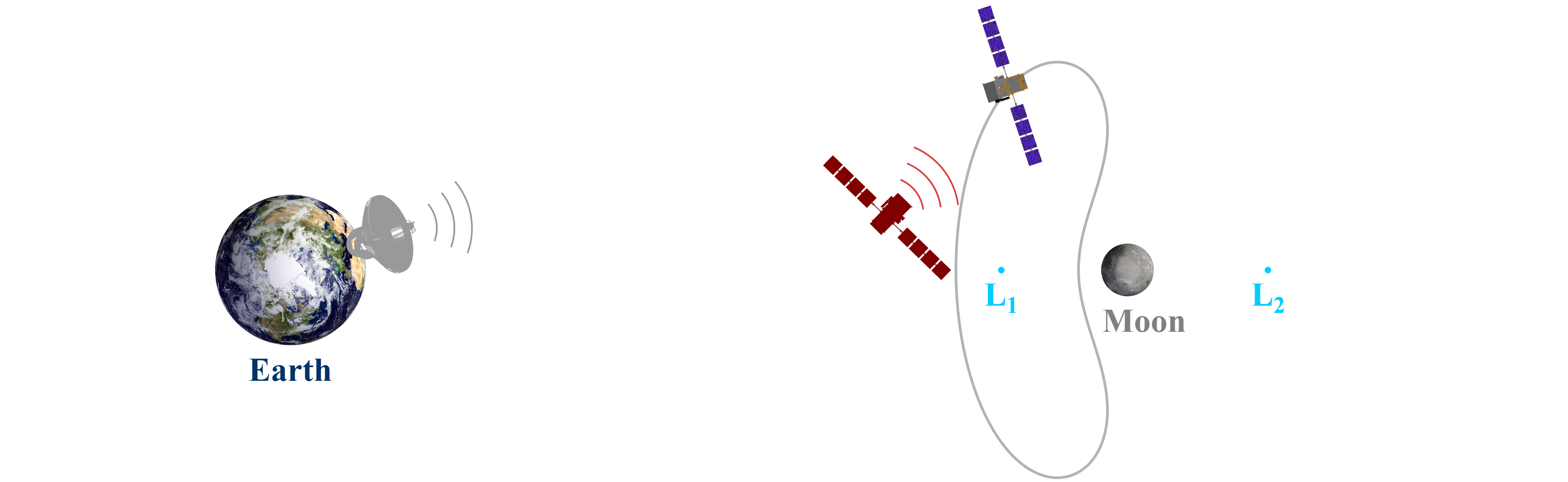}
    \caption{Close-proximity jamming in a periodic orbit around $L_1$. Because of the vast distances involved, the Earth-satellite communication link is relatively weak, making it especially susceptible to jamming. To mitigate the jamming effect, the satellite must maneuver away from the jammer. }
    \label{fig:cislunar_jamming_diagram}
\end{figure}

We address an evader-pursuer interaction in cislunar space by formulating it as a zero-sum differential game. The evader seeks to remain close to its nominal orbit while avoiding the pursuer, whose objective is to approach and potentially disrupt communications. The game is posed in the circular restricted three-body problem and incorporates features critical for navigating cislunar space: autonomous phasing adjustment of the reference orbit to enable aggressive maneuvers, and inclusion of stable and unstable manifolds in the evader’s cost to encourage fuel-efficient motion. We solve the nonlinear game using continuous-time differential dynamic programming, obtaining local saddle-point feedback strategies that balance separation, reference tracking, and fuel efficiency. These results advance defensive guidance methods to secure future cislunar operations against physical and cyber threats.

\textbf{Contributions.} 
This paper makes the following contributions.
\begin{itemize}
\item We extend \emph{pursuit-evasion games to the cislunar regime} by explicitly modeling the nonlinear dynamics of the circular restricted three-body problem.
\item We introduce \emph{reference orbit phasing adjustment} as an additional control variable to enable along-track evasive maneuvers not captured in prior Keplerian formulations.
\item We design a \emph{manifold-aware cost function} that embeds stable and unstable orbital directions, guiding evading strategies toward fuel-efficient and lower-risk maneuvers.
\item We benchmark against a linear-quadratic pursuer, showing that our nonlinear formulation, solved via continuous-time differential dynamic programming, outperforms simple quadratic approximations.
\end{itemize}

\textbf{Notation.}  
$I_n$ denotes an identity matrix of dimension $n\times n$. $0_{n}$ denotes a null matrix of dimension $n\times n$. 
$\norm{x}$ is the $\ell_2$ norm of vector or function $x$, whereas
$\norm{x}_Q$ is the $Q-$weighted $\ell_2$ norm of vector or function $x$. Subscripts of the functions $V,\phi,\bar{V},\bar{L},\bar{F}$ denote derivatives with respect to the indicated variable, for example, $V_\mf{x} = \nabla_\mf{x} V$, $V_{\mf{x}\mf{x}} = \nabla_\mf{x}^2 V$, $\bar{L}_\mf{w} = \nabla_\mf{w} \bar{L}$.

\section{Problem Formulation}
\subsection{Spacecraft Dynamics in Cislunar Space}

In cislunar space, one can model the dynamics of a spacecraft under the gravity of the Earth and the Moon according to the circular restricted three‐body problem. In this model, the Earth with mass $m_e$ and the Moon with mass $m_m$ move in circular orbits around their barycenter, while the spacecraft moves under the combined Earth-Moon gravitational field. To describe the spacecraft's motion, let us denote the mass ratio $\mu = \frac{m_m}{m_e + m_m}$ and normalize the system's masses and length units, so that $m_e=1-\mu$, $m_m=\mu$, and so that the distance between the Earth and the Moon is equal to one. In addition, let us consider the rotating frame with origin at the three-body system's barycenter, in which the Earth and the Moon are located at $p_e=[-\mu~0~0]^\mathrm{T}$ and $p_m=[1-\mu~0~0]^\mathrm{T}$. Then, the rotating, non-dimensional equations of motion for the spacecraft in this frame are given by \cite{koon2000dynamical}: 
\begin{equation}\label{eq:CR3BP}
\begin{split}
\ddot{x}&=2\dot{y}+x-\frac{(1-\mu)(x+\mu)}{r_e^3}-\frac{\mu}{r_m^3}(x-1+\mu)+\frac{u_x}{m},\\
\ddot{y}&=-2\dot{x}+y-\frac{1-\mu}{r_e^3}y-\frac{\mu}{r_m^3}y+\frac{u_y}{m},\\
\ddot{z}&=-\frac{1-\mu}{r_e^3}z-\frac{\mu}{r_m^3}z+\frac{u_z}{m}.
\end{split}
\end{equation}
Here, $(x, y, z)$, $(\dot{x}, \dot{y}, \dot{z})$, $(u_x,u_y,u_z)$ are the spacecraft's normalized position, velocity, and thrust, $m$ is its mass, and
\begin{equation*}
r_e=\sqrt{(x+\mu)^2+y^2+z^2},\qquad r_m=\sqrt{(x-1+\mu)^2+y^2+z^2},
\end{equation*}
are the normalized distances of the spacecraft from the Earth and the Moon. Denoting $\mathbf{x}=[x ~ y ~ z ~ \dot{x} ~ \dot{y} ~ \dot{z}]^\mathrm{T}\in\mathbb{R}^6$ and $\mathbf{u}=[u_x ~ u_y ~ u_z]^\mr{T}\in\mathbb{R}^3$, we can rewrite the spacecraft's dynamics \eqref{eq:CR3BP} in the compact form
\begin{equation}\label{eq:CR3BPc}
\dot{\mathbf{x}}=f(\mathbf{x})+B\mathbf{u},
\end{equation}
where $B = [0_{3} ~ \frac{1}{m}I_3 ]^\textrm{T}$ and
\begin{equation}\label{eq:CR3BPc2}
\begin{split}
f(\mathbf{x}):=\begin{bmatrix}f_x(\mathbf{x}) \\ f_y(\mathbf{x}) \\f_z(\mathbf{x}) \\f_{\dot{x}}(\mathbf{x}) \\f_{\dot{y}}(\mathbf{x}) \\f_{\dot{z}}(\mathbf{x}) \end{bmatrix}=\begin{bmatrix}\dot{x} \\ \dot{y} \\ \dot{z} \\ 2\dot{y}+x-\frac{(1-\mu)(x+\mu)}{r_e^3}-\frac{\mu}{r_m^3}(x-1+\mu) \\ -2\dot{x}+y-\frac{1-\mu}{r_e^3}y-\frac{\mu}{r_m^3}y \\ -\frac{1-\mu}{r_e^3}z-\frac{\mu}{r_m^3}z\end{bmatrix}.
\end{split}
\end{equation}

There are five equilibrium points $L_i$, $i=1,\ldots,5$, in the Earth-Moon orbital plane where centrifugal and gravitational forces exactly cancel, and which correspond to fixed points of \eqref{eq:CR3BPc} under $\mathbf{u}=0$. These are called the Lagrange points, and a particle placed at one of these with zero initial velocity will remain there (in theory) indefinitely. From the perspective of lunar missions, Lagrange points and their neighboring orbits are interesting because they are relatively invariant locations for placing scientific instruments, which we can use for communication relays, refueling, astronomy, and other purposes. However, both $L_1$ and $L_2$ points are inherently unstable. Any spacecraft placed into a periodic or quasi-periodic orbit around either point must carry out regular stationkeeping maneuvers. Nevertheless, for $L_2$, the payoff is continuous, unobstructed visibility of both Earth and the lunar far side, a capability no low lunar orbit can match. China’s Queqiao-1 spacecraft has occupied a halo orbit around $L_2$ since 2018, relaying Chang’e-4’s far-side communications \cite{li2021overview}. More recently, NASA’s CAPSTONE mission entered a near-rectilinear halo orbit about $L_2$ to map out the stability regime that the future Lunar Gateway will exploit \cite{advancedspace_news2023}. These examples demonstrate that Lagrange-point orbits offer unique communications and observational advantages for cislunar exploration.

\subsection{Adversarial Pursuits in Cislunar Space}

Despite having several desirable properties, cislunar orbits present significant security issues. One of these is their lack of redundancy and remoteness, which lead to high mission costs and communication challenges. An adversary in these orbits can exploit such vulnerabilities by jamming co-located relays or interfering with optical sensors, potentially interrupting communications between the Earth and the Moon. Unlike operations in low-Earth orbits, where redundancy is plentiful and where satellites often operate in constellations, spacecraft in cislunar orbits may not be able to deal with jamming by means of rerouting. Instead, they must search for novel cyber-physical defense methods that take the limitations of lunar missions explicitly into account.

A potential defense mechanism against jamming is to perform evasive maneuvers. Such maneuvers can increase the distance from an adversarial jammer and hence attenuate the effect of their interference. Many studies have also examined them in depth in the context of near-Earth missions \cite{fu2025analytical, mehlman2024cat}. Nevertheless, orbits in cislunar space are vastly different than those near Earth; they are unstable, and their evolution is dictated by the CR3BP dynamics \eqref{eq:CR3BP}. Spacecraft, relays, and other scientific instruments in cislunar orbits must execute evasive maneuvers with care: while evading, they must remain close to the stability provided by their nominal orbit.
Otherwise, they risk excessive fuel consumption, prolonged deviations lasting weeks, or even permanent loss of the mission due to uncontrolled drift.

In view of the above, we formulate a pursuit-evasion interaction over the CR3BP dynamics \eqref{eq:CR3BP}, where both the pursuer and the evader aim to control their separation while remaining close to the relative stability of their nominal cislunar orbit. Specifically, let us consider an evader spacecraft with state $\mathbf{x}_\mathrm{e}=[\mf{p}_\mathrm{e}^\mr{T}~\mf{v}_\mathrm{e}^\mr{T}]^\mathrm{T}\in\mathbb{R}^6$, and a pursuer spacecraft with state $\mathbf{x}_\mathrm{p}=[\mf{p}_\mathrm{p}^\mr{T}~\mf{v}_\mathrm{p}^\mr{T}]^\mathrm{T}\in\mathbb{R}^6$.  From \eqref{eq:CR3BPc}, the translational dynamics of these spacecraft will evolve according to the CR3BP equations
\begin{equation}\label{eq:spacecraftdyn}
\begin{split}
\dot{\mathbf{x}}_\mathrm{e}(t)&=f({\mathbf{x}}_\mathrm{e}(t))+B_\mr{e}\mathbf{u}_\mathrm{e}(t),\quad {\mathbf{x}}_\mathrm{e}(t_0)={\mathbf{x}}_\mathrm{e0},\\
\dot{\mathbf{x}}_\mathrm{p}(t)&=f({\mathbf{x}}_\mathrm{p}(t))+B_\mr{p}\mathbf{u}_\mathrm{p}(t),\quad {\mathbf{x}}_\mathrm{p}(t_0)={\mathbf{x}}_\mathrm{p0},
\end{split}
\end{equation}
where $B_\mr{e}=[0_{3}~\frac{1}{m_\mr{e}}I_3]^\mr{T}$, $B_\mr{p}=[0_{3}~\frac{1}{m_\mr{p}}I_3]^\mr{T}$, with $m_\mr{e},~m_\mr{p}$ denoting the masses of the evader and pursuer, respectively. The control inputs $\mathbf{u}_\mathrm{e}, \mathbf{u}_\mathrm{p}\in\mathbb{R}^3$ are the normalized thrusts of the evader and the pursuer. 

The purpose of the evader is to follow a desired reference orbit $\mathbf{x}_\mr{de}:\mathbb{R}_+\rightarrow\mathbb{R}^6$ while increasing its distance from the pursuer. On the other hand, the purpose of the pursuer is to track a desired reference orbit $\mathbf{x}_\textrm{dp}:\mathbb{R}_+\rightarrow\mathbb{R}^6$ while decreasing its distance from the evader. We assume both spacecraft follow the same cislunar orbit, and hence we could write $\mathbf{x}_\textrm{de} = \mathbf{x}_\textrm{dp} = \mathbf{x}_\textrm{d}$, where $\mathbf{x}_\textrm{d} : \mathbb{R}_+ \rightarrow \mathbb{R}^6$ denotes a moving reference point along the orbit. However, since the spacecraft occupy different positions on the orbit, it is more accurate to state that
\begin{equation}\label{eq:refs}
\begin{split}
\mathbf{x}_\textrm{de}(t)&=\mathbf{x}_\textrm{d}(c_\mr{e}(t)),\\
\mathbf{x}_\textrm{dp}(t)&=\mathbf{x}_\textrm{d}(c_\mr{p}(t)),
\end{split}
\end{equation}
where $c_\mr{e}(t)$ and $c_\mr{p}(t)$ indicate the different phases of the spacecraft along the orbit.

\begin{remark}
While most adversarial pursuits focus on controlling the relative distance between the evader and the pursuer and ignore reference tracking tasks, these tasks cannot be dispensed with in cislunar space. Since cislunar orbital dynamics are unstable, spacecraft must be bound to follow the relatively stable manifold in the vicinity of their reference orbit, or risk mission loss due to uncontrolled drift. Moreover, large control impulses that change the Jacobi constant enough to cross the critical values at $L_1$ or $L_2$ will open (or close) the necks of the zero-velocity surface, changing the Hill’s regions and thus reachability.
\end{remark} % Add that an increase in energy could lead to new space opened

In what follows, we formulate a nonlinear game over the abovementioned CR3BP dynamics, enabling aggressive yet fuel-efficient evasion maneuvers that account for the intrinsic orbital geometry of cislunar motion.

\section{Zero-Sum Dynamic Game for Cislunar Adversarial Pursuits}

In this section, we formulate a nonlinear pursuit-evasion game in cislunar space that incorporates crucial aspects of the full CR3BP. This includes autonomous adjustment of the reference orbit phasing to enable efficient evading maneuvers, and inclusion of stable and unstable manifolds in the cost to account for fuel-efficient maneuvers.

\subsection{Reference Orbit Phasing Adjustment}

When the reference signals $\mathbf{x}_\textrm{de}, \mathbf{x}_\textrm{dp}$ are fixed, the evading and the pursuing maneuvers of the two spacecraft become substantially restricted. This is because the spacecraft are forced to confine themselves in a vicinity of $\mathbf{x}_\textrm{de}, \mathbf{x}_\textrm{dp}$, whose prescribed temporal evolution restricts effective evasion and pursuit maneuvers. Instead, the spacecraft should be able to use thrust to complete the orbit either faster or slower than normal, hence enabling maneuvers \textit{along} the orbit rather than about it.

To enable maneuvers along the cislunar orbit, let us define the following dynamically controllable phases of the spacecraft reference signals:
\begin{equation}\label{eq:taudyn}
\begin{split}
\dot{c}_\mr{e}(t)&=\tau_\mr{e}(t), \quad {c}_\mr{e}(t_0)=t_{0\mr{e}},\\
\dot{c}_\mr{p}(t)&=\tau_\mr{p}(t), \quad {c}_\mr{p}(t_0)=t_{0\mr{p}},
\end{split}
\end{equation}
where $\tau_\mr{e}(t),~\tau_\mr{p}(t)$ are phase control variables that enable the spacecraft to scale the evolution of their reference orbit phasing. Note that when $\tau_\mr{e}(t),\tau_\mr{p}(t)>1$, then the desired reference signals of the spacecraft \eqref{eq:refs} evolve more quickly than $\mathbf{x}_\mr{d}(t)$. This indicates that the spacecraft want to ``speed through'' the cislunar orbit, and enables them to perform more flexible evading and pursuing maneuvers. Vice versa, when $\tau_\mr{e}(t),\tau_\mr{p}(t)<1$, then the spacecraft ``slow down'' along the trajectory and take longer than the nominal traversal time. With these new variables that control the reference orbit phasing, 
 the spacecraft policies now constitute the tuples $\mathbf{w}_\mr{e}=[\mathbf{u}_\mr{e}^\textrm{T}~\tau_\mr{e}]^\textrm{T}\in\mathbb{R}^4$ and $\mathbf{w}_\mr{p}=[\mathbf{u}_\mr{p}^\textrm{T}~\tau_\mr{p}]^\textrm{T}\in\mathbb{R}^4$, instead of simply the thrusts $\mathbf{u}_\mr{e}$ and $\mathbf{u}_\mr{p}$.

\subsection{Nonlinear Dynamic Game Formulation}

Given the expanded strategy spaces of the evader and the pursuer, we now define the pursuit-evasion game over the CR3BP dynamics. To this end, we note that the reference phasing controls \eqref{eq:taudyn} should typically remain close to unity, since significant deviations can induce excessive stationkeeping costs and potentially unstable maneuvers. With this observation in place, and defining the concatenated state $\mathbf{x}=[\mathbf{x}_\mathrm{e}^\mr{T}~\mathbf{x}_\mathrm{p}^\mr{T}~c_\mr{e}~c_\mr{p}]^\mr{T}\in\mathbb{R}^{14}$, we formulate the cost function of the pursuit-evasion game under the dynamics \eqref{eq:spacecraftdyn} and \eqref{eq:taudyn} as 
\begin{align}\label{eq:nlgame}
\min_{\mathbf{w}_\mr{e}} \max_{\mathbf{w}_\mr{p}} J(\mathbf{w}_\mathrm{e}, \mathbf{w}_\mathrm{p} )=\int_{t_0}^{t_f} L(\mathbf{x}(t),\mathbf{w}_\mr{e}(t),\mathbf{w}_\mr{p}(t),t)\mathrm{d}t + \phi(\mathbf{x}(t_f), t_f)
\end{align}
where
\begin{align*}
L&:= \norm{\Delta{\mathbf{x}}_\mathrm{e}(t)}_{Q_\mr{e}(t)}^2
+ \norm{\mathbf{u}_\mathrm{e}(t)}_{R_\mathrm{e}(t)}^2  + a_\mr{e}(t)(\tau_\mr{e}(t)-1)^2  - \norm{\Delta{\mathbf{x}}_\mathrm{p}(t)}_{Q_\mathrm{p}(t)}^2 - \norm{\mathbf{u}_\mathrm{p}(t)}_{R_\mathrm{p}(t)}^2 - a_\mr{p}(t)(\tau_\mr{p}(t)-1)^2 \\&\qquad\qquad\qquad\qquad\qquad\qquad\qquad\qquad\qquad\qquad\qquad\qquad\qquad\qquad\qquad\qquad\qquad\qquad+S\left(\norm{\mf{p}_{\mr{e}}(t)-\mf{p}_{\mr{p}}(t)}\right), \\  \phi&:=\norm{\Delta{\mathbf{x}}_\mathrm{e}(t_f)}_{F_\mathrm{e}}^2 - \norm{\Delta{\mathbf{x}}_\mathrm{p}(t_f)}_{F_\mr{p}}^2+S\left(\norm{\mf{p}_{\mr{e}}(t_f)-\mf{p}_{\mr{p}}(t_f)}\right).
\end{align*}
This game is subject to
\begin{align}\label{eq:nonlineardyn}
\dot{\mf{x}}(t)=F(\mf{x},\mf{w}_\mr{e},\mf{w}_\mr{p}):=\begin{bmatrix} 
f({\mathbf{x}}_\mathrm{e}(t))+B_\mr{e}\mathbf{u}_\mathrm{e}(t)  \\
f({\mathbf{x}}_\mathrm{p}(t))+B_\mr{p}\mathbf{u}_\mathrm{p}(t) \\ \tau_\mr{e}(t) \\ 
\tau_\mr{p}(t)\end{bmatrix}, \quad\mf{x}(t_0)=\mf{x}_0:=\begin{bmatrix}{\mathbf{x}}_\mathrm{e0}  \\  {\mathbf{x}}_\mathrm{p0} \\ t_\mr{0e}  \\ t_\mr{0p} \end{bmatrix},
\end{align}
with $Q_\textrm{e},~R_\textrm{e}, ~ F_\textrm{e},~Q_\textrm{p},~R_\textrm{p}, ~ F_\textrm{p}\succ0$, $a_\mr{e},~ a_\mr{p}>0$ being weighting matrices and scalars, and $\Delta\mathbf{x}_\textrm{i}(t)=\mathbf{x}_\textrm{i}(t)-\mathbf{x}_\mathrm{di}(t)=\mathbf{x}_\textrm{i}(t)-\mathbf{x}_\mathrm{d}(c_\mr{i}(t))$, $\mr{i}\in\{\mathrm{e}, \mr{p} \}$, being the orbital tracking errors.

In the cost of the game \eqref{eq:nlgame}, the terms $\|\Delta{\mathbf{x}}_\mathrm{e}(t)\|_{Q_\mathrm{e}(t)}^2$, $\|\Delta{\mathbf{x}}_\mathrm{e}(t_f)\|_{F_\mathrm{e}}^2$
incentivize the evader to remain close to the nominal orbit $\mf{x}_\mr{de}$, and $\|\mathbf{u}_\textrm{e}(t)\|_{R_\mathbf{e}(t)}^2$ captures the requirement that fuel consumption is minimal. Moreover, $a_\mr{e}(t)(\tau_\mr{e}(t)-1)^2$ forces the reference phasing controls to remain close to unity.
Finally, $S: \mathbb{R}_+ \rightarrow \mathbb{R}_+$ is a function that penalizes the evader when it is within a specified radius $d_0$ of the pursuer and converges to zero otherwise. A relevant choice for the function $S$ is
\begin{equation}\label{eq:evadecost}
S(d)=\begin{cases}\frac{1}{p}w(d_0-d)^p,~ &d\le d_0 \\ 0   ~&d\ge d_0 \end{cases},
\end{equation}
with $w>0$ and $p>2$. Note that this function is twice continuously differentiable and thus useful in our framework. The rest of the terms reflect reciprocal costs for the pursuer. 

We are particularly interested in a tuple of policies $\{\mathbf{w}_\mathrm{e}^\star, \mathbf{w}_\mathrm{p}^\star\}$ that is a saddle-point solution to the game \eqref{eq:nlgame}. In other words, this tuple should satisfy the inequality
\begin{equation*}
J(\mathbf{w}_\mathrm{e}^\star, \mathbf{w}_\mathrm{p}) \le J(\mathbf{w}_\mathrm{e}^\star, \mathbf{w}_\mathrm{p}^\star) \le J(\mathbf{w}_\mathrm{e}, \mathbf{w}_\mathrm{p}^\star ),\quad \forall \mathbf{w}_\mathrm{e}, \mathbf{w}_\mathrm{p}.
\end{equation*}
To find such a tuple, one needs to compute the value function 
$V:\mathbb{R}^{14}\times[t_0,~t_f]\rightarrow\mathbb{R}_+$ of the game \cite{bacsar1998dynamic}. 
This function is defined as
\begin{equation}
    V(\mathbf{x}_0,t_0) \;=\; J(\mathbf{w}_\mathrm{e}^\star, \mathbf{w}_\mathrm{p}^\star),
\end{equation}
where $\mathbf{x}_0 = 
[  {\mathbf{x}}_\mathrm{e0}^\mr{T}  ~ 
    {\mathbf{x}}_\mathrm{p0}^\mr{T} ~ t_\mr{0e} ~ t_\mr{0p}
]^\mr{T}$
is the concatenated vector of initial states. It is obtained as the unique viscosity solution of the Hamilton--Jacobi--Isaacs (HJI) partial differential equation:
\begin{equation}\label{eq:HJI}
-\frac{\partial V(\mathbf{x},t)}{\partial t} = \min_{\mathbf{w}_\mr{e}} \max_{\mathbf{w}_\mr{p}} \Big\{ L(\mathbf{x},\mathbf{w}_\mr{e},\mathbf{w}_\mr{p},t) + V_{\mathbf{x}}^\mr{T}(\mathbf{x},t)F(\mf{x},\mf{w}_\mr{e},\mf{w}_\mr{p})\Big\}, \quad V(\mathbf{x}, t_f)=\phi(\mathbf{x}, t_f),
\end{equation}
where $V_\mf{x}$ is the gradient of $V$ with respect to $\mf{x}$.
Note here that the right-hand side of \eqref{eq:HJI} is strictly convex in $\mathbf{w}_e$, strictly concave in $\mathbf{w}_p$, and separable. Therefore, the min and the max operators in \eqref{eq:HJI} can be interchanged and yield the same result, meaning that Isaacs' condition is met. Still, \eqref{eq:HJI} remains difficult to solve analytically in general, as it is nonlinear and high-dimensional. For this reason, we will employ a Differential Dynamic Programming (DDP) algorithm that solves it locally through iterative linear-quadratic approximations.

\subsection{Game-Theoretic Differential Dynamic Programming}

Most existing DDP algorithms that approximate Hamilton-Jacobi equations through 
iterative linear-quadratic expansions rely on discretized system dynamics, as 
discretization facilitates the use of matrix operations and simplifies computations. 
However, in the CR3BP setting, the three-body dynamics are unstable, planning horizons span days, and 
discretization errors accumulate quickly even with small step sizes. This issue 
is particularly pronounced for near-rectilinear halo orbits with close perilune passages, 
where the dynamics vary rapidly. In practice, applying discrete-time DDP would not only suffer from error accumulation but would also require very fine discretization, resulting in increased memory and computational demands. For this 
reason, we adopt a continuous-time formulation of DDP that allows for the use of 
adaptive step-size solvers, following~\cite{sun2018min}.

\renewcommand\thefootnote{\arabic{footnote}}

DDP proceeds by locally expanding the HJI equation \eqref{eq:HJI} about a nominal trajectory $(\bar{\mathbf{x}},~\bar{\mathbf{w}}_\mr{e},~\bar{\mathbf{w}}_\mr{p})$. 
This expansion is quadratic in the value function and linear in the system dynamics. 
We denote nominal trajectories with a bar, e.g., $\bar{V}=V(\bar{\mathbf{x}},t)$, and define the perturbations
\begin{equation*}\vspace{-0.3mm}
\delta\mathbf{x}=\mathbf{x}-\bar{\mathbf{x}}, \quad 
\delta\mathbf{w}_{\mr{e}}=\mathbf{w}_{\mr{e}}-\bar{\mathbf{w}}_{\mr{e}}, \quad 
\delta\mathbf{w}_{\mr{p}}=\mathbf{w}_{\mr{p}}-\bar{\mathbf{w}}_{\mr{p}}.
\end{equation*}
A first-order approximation of the perturbed dynamics then takes the form
\begin{equation}\label{eq:forward}
\frac{\mr{d}\delta\mf{x}}{\mr{d}t} \approx \bar{F}_{\mf{x}}\delta\mf{x}+\bar{F}_{\mf{w}_{\mr{e}}}\delta\mf{w}_{\mr{e}}+\bar{F}_{\mf{w}_{\mr{p}}}\delta\mf{w}_{\mr{p}}.
\end{equation}
Subsequently, a second-order expansion of the left-hand side of \eqref{eq:HJI}  about $(\bar{\mathbf{x}},~\bar{\mathbf{w}}_\mr{e},~\bar{\mathbf{w}}_\mr{p})$ is
\begin{equation}\label{eq:left_exp} 
\begin{split}
-\frac{\partial V(\mathbf{x},t)}{\partial t} &\approx -\frac{\partial\bar{V}}{\partial t} -\frac{\partial\bar{V}_{\mathbf{x}}^\mr{T}}{\partial t}\delta\mathbf{x} - \frac{1}{2}\delta\mf{x}^\mr{T}\frac{\partial\bar{V}_{\mf{x}\mf{x}}}{\partial t}\delta\mf{x} \\ &=-\frac{\mr{d}\bar{V}}{\mr{d} t} -\frac{\mr{d}\bar{V}_{\mathbf{x}}^\mr{T}}{\mr{d} t}\delta\mathbf{x} - \frac{1}{2}\delta\mf{x}^\mr{T}\frac{\mr{d}\bar{V}_{\mf{x}\mf{x}}}{\mr{d} t}\delta\mf{x} + \bar{V}_\mf{x}^\textrm{T}\bar{F} + \delta\mf{x}^\mr{T}\bar{V}_{\mf{x}\mf{x}}\bar{F}+\frac{1}{2}\delta\mf{x}^\mr{T}\sum_{i=1}^{{14}}\bar{V}_{\mf{x}\mf{x}\mf{x}}^{(i)}\bar{F}^{(i)}\delta\mf{x},
\end{split}
\end{equation}
where $\bar{V}_{\mathbf{x}\mathbf{x}\mathbf{x}}^{(i)}$ denotes the Hessian of the $i$-th entry of $\bar{V}_\mf{x}$, and $\bar{F}^{(i)}$ denotes the $i$-th entry of $\bar{F}$. In addition, an expansion of the right-hand side of \eqref{eq:HJI}  about $(\bar{\mathbf{x}},~\bar{\mathbf{w}}_\mr{e},~\bar{\mathbf{w}}_\mr{p})$  yields
\begin{equation}\label{eq:right_exp}
\begin{split}
\min_{\mathbf{w}_\mr{e}} \max_{\mathbf{w}_\mr{p}} \Big\{ & L(\mathbf{x},\mathbf{w}_\mr{e},\mathbf{w}_\mr{p},t) + V_{\mathbf{x}}^\mr{T}(\mathbf{x},t)F(\mf{x},\mf{w}_\mr{e},\mf{w}_\mr{p})\Big\} \\\approx &\min_{\mathbf{w}_\mr{e}} \max_{\mathbf{w}_\mr{p}} \Big\{\bar{L}+\bar{L}_\mf{x}^\mr{T}\delta\mf{x}+\bar{L}_{\mf{w}_\mr{e}}^\mr{T}\delta{\mf{w}_\mr{e}}+\bar{L}_{\mf{w}_\mr{p}}^\mr{T}\delta{\mf{w}_\mr{p}} + \frac{1}{2}\begin{bmatrix}\delta\mf{x} \\ \delta\mf{w}_\mr{e} \\ \delta\mf{w}_\mr{p} \end{bmatrix}^\mr{T} \begin{bmatrix} \bar{L}_{\mf{x}\mf{x}} & \bar{L}_{\mf{x}\mf{w}_\mr{e}} & \bar{L}_{\mf{x}\mf{w}_\mr{p}} \\ \bar{L}_{\mf{w}_\mr{e}\mf{x}} & \bar{L}_{\mf{w}_\mr{e}\mf{w}_\mr{e}} & \bar{L}_{\mf{w}_\mr{e}\mf{w}_\mr{p}} \\ \bar{L}_{\mf{w}_\mr{p}\mf{x}} & \bar{L}_{\mf{w}_\mr{p}\mf{w}_\mr{e}} & \bar{L}_{\mf{w}_\mr{p}\mf{w}_\mr{p}}  \end{bmatrix}    \begin{bmatrix}\delta\mf{x} \\ \delta\mf{w}_\mr{e} \\ \delta\mf{w}_\mr{p} \end{bmatrix} \\ & + \bar{V}_{\mf{x}}^\mr{T}\bar{F} +  \bar{V}_{\mf{x}}^\mr{T}\bar{F}_\mf{x}\delta\mf{x}+  \bar{V}_{\mf{x}}^\mr{T}\bar{F}_{\mf{w}_\mr{e}}\delta\mf{w}_\mr{e}+  \bar{V}_{\mf{x}}^\mr{T}\bar{F}_{\mf{w}_\mr{p}}\delta\mf{w}_\mr{p} + \delta \mf{x}^\mr{T}\bar{V}_{\mf{x}\mf{x}}\bar{F} \\& + \delta \mf{x}^\mr{T}\bar{V}_{\mf{x}\mf{x}}\bar{F}_\mf{x}\delta\mf{x} + \delta \mf{x}^\mr{T}\bar{V}_{\mf{x}\mf{x}}\bar{F}_{\mf{w}_\mr{e}}\delta\mf{w}_\mr{e} + \delta \mf{x}^\mr{T}\bar{V}_{\mf{x}\mf{x}}\bar{F}_{\mf{w}_\mr{p}}\delta\mf{w}_\mr{p} + \frac{1}{2}\delta\mf{x}^\mr{T} \sum_{i=1}^{14}\bar{V}_{\mf{x}\mf{x}\mf{x}}^{(i)}\bar{F}^{(i)} \delta\mf{x} \Big\}.
\end{split}
\end{equation}
In our setting, the cost function $L$ is separable in $\mathbf{x}$, $\mathbf{w}_\mr{e}$ and $\mathbf{w}_\mr{p}$, hence we have $\bar{L}_{\mathbf{x}\mathbf{w}_\mr{e}}=\bar{L}_{\mathbf{w}_\mr{e}\mathbf{x}}=\bar{L}_{\mathbf{x}\mathbf{w}_\mr{p}}=\bar{L}_{\mathbf{w}_\mr{p}\mathbf{x}}=\bar{L}_{\mathbf{w}_\mr{p}\mathbf{w}_\mr{e}}=\bar{L}_{\mathbf{w}_\mr{e}\mathbf{w}_\mr{p}}=0$. Given this, equating \eqref{eq:left_exp} to \eqref{eq:right_exp} we obtain
\begin{multline}\label{eq:HJIapprox}
-\frac{\mr{d}\bar{V}}{\mr{d} t} -\frac{\mr{d}\bar{V}_{\mathbf{x}}^\mr{T}}{\mr{d} t}\delta\mathbf{x} - \frac{1}{2}\delta\mf{x}^\mr{T}\frac{\mr{d}\bar{V}_{\mf{x}\mf{x}}}{\mr{d} t}\delta\mf{x} = \min_{\mathbf{w}_\mr{e}} \max_{\mathbf{w}_\mr{p}} \Big\{ \bar{L} + \delta\mf{x}^\mr{T}\bar{Q}_{\mf{x}}+ \delta\mf{w}_\mr{e}^\mr{T}\bar{Q}_{\mf{w}_\mr{e}}+ \delta\mf{w}_\mr{p}^\mr{T}\bar{Q}_{\mf{w}_\mr{p}}\\ + \frac{1}{2}\delta\mf{x}^\mr{T}\bar{Q}_{\mf{x}\mf{x}}\delta\mf{x} + \frac{1}{2}\delta\mf{w}_\mr{e}^\mr{T}\bar{Q}_{\mf{w}_\mr{e}\mf{w}_\mr{e}}\delta\mf{w}_\mr{e} + \frac{1}{2}\delta\mf{w}_\mr{p}^\mr{T}\bar{Q}_{\mf{w}_\mr{p}\mf{w}_\mr{p}}\delta\mf{w}_\mr{p} + \delta\mf{w}_\mr{e}^\mr{T} \bar{Q}_{\mf{w}_\mr{e}\mf{x}}\ \delta\mf{x}+ \delta\mf{w}_\mr{p}^\mr{T} \bar{Q}_{\mf{w}_\mr{p}\mf{x}}\ \delta\mf{x} \Big\}
\end{multline}
where
\begin{subequations}
\begin{align*}
&\bar{Q}_{\mf{x}}      = \bar{F}_{\mf{x}}^\mr{T}\bar{V}_{\mf{x}} + \bar{L}_{\mf{x}}, 
&&\bar{Q}_{\mf{w}_{\mr{e}}} = \bar{F}_{\mf{w}_{\mr{e}}}^\mr{T}\bar{V}_{\mf{x}} + \bar{L}_{\mf{w}_{\mr{e}}}, 
&&\bar{Q}_{\mf{w}_{\mr{p}}} = \bar{F}_{\mf{w}_{\mr{p}}}^\mr{T}\bar{V}_{\mf{x}} + \bar{L}_{\mf{w}_{\mr{p}}}, \\
&\bar{Q}_{\mf{x}\mf{x}} = \bar{L}_{\mf{x}\mf{x}} + \bar{V}_{\mf{x}\mf{x}}\bar{F}_{\mf{x}} + \bar{F}_{\mf{x}}^\mr{T}\bar{V}_{\mf{x}\mf{x}}, 
&&\bar{Q}_{\mf{w}_\mr{e}\mf{w}_\mr{e}} = \bar{L}_{\mf{w}_\mr{e}\mf{w}_\mr{e}}, 
&&\bar{Q}_{\mf{w}_\mr{p}\mf{w}_\mr{p}} = \bar{L}_{\mf{w}_\mr{p}\mf{w}_\mr{p}},\\
& \bar{Q}_{\mf{w}_\mr{e}\mf{x}}=\bar{F}_{\mf{w}_\mr{e}}^\mr{T}\bar{V}_{\mf{x}\mf{x}}, && \bar{Q}_{\mf{w}_\mr{p}\mf{x}}=\bar{F}_{\mf{w}_\mr{p}}^\mr{T}\bar{V}_{\mf{x}\mf{x}}.
\end{align*}
\end{subequations}
Using the first-order stationary condition to compute the min and max in \eqref{eq:HJIapprox} yields the optimal controls
\begin{equation}\label{eq:updates}
\delta\mf{w}_{\mr{e}}^\star = \ell_{\mf{w}_{\mr{e}}} + \mf{K}_{\mf{w}_{\mr{e}}}\delta\mf{x}, \qquad \delta\mf{w}_{\mr{p}}^\star = \ell_{\mf{w}_{\mr{p}}} + \mf{K}_{\mf{w}_{\mr{p}}}\delta\mf{x},
\end{equation}
where
\begin{equation}\label{eq:gains}
\begin{split}
   &\ell_{\mf{w}_\mr{e}}=-\bar{Q}_{\mf{w}_\mr{e}\mf{w}_\mr{e}}^{-1}\bar{Q}_{\mf{w}_\mr{e}},\quad \mf{K}_{{\mf{w}_\mr{e}}}=-\bar{Q}_{\mf{w}_\mr{e}\mf{w}_\mr{e}}^{-1}\bar{Q}_{\mf{w}_\mr{e}\mf{x}},\\
   &\ell_{\mf{w}_\mr{p}}=-\bar{Q}_{\mf{w}_\mr{p}\mf{w}_\mr{p}}^{-1}\bar{Q}_{\mf{w}_\mr{p}},\quad \mf{K}_{{\mf{w}_\mr{p}}}=-\bar{Q}_{\mf{w}_\mr{p}\mf{w}_\mr{p}}^{-1}\bar{Q}_{\mf{w}_\mr{p}\mf{x}}.
   \end{split}
\end{equation}
Finally, equating the zero, first, and second order terms in \eqref{eq:HJIapprox} under the controls \eqref{eq:updates} yields the backward differential equations
\begin{equation}\label{eq:backward}
\begin{split}
-\frac{\mr{d}\bar{V}}{\mr{d}t} &= \bar{L}+\ell_{\mf{w}_\mr{e}}^\textrm{T}\bar{Q}_{\mf{w}_\mr{e}}+\ell_{\mf{w}_\mr{p}}^\textrm{T}\bar{Q}_{\mf{w}_\mr{p}} + \frac{1}{2}\ell_{\mf{w}_\mr{e}}^\textrm{T}\bar{Q}_{\mf{w}_\mr{e}\mf{w}_\mr{e}}\ell_{\mf{w}_\mr{e}} + \frac{1}{2}\ell_{\mf{w}_\mr{p}}^\textrm{T}\bar{Q}_{\mf{w}_\mr{p}\mf{w}_\mr{p}}\ell_{\mf{w}_\mr{p}}, \\
-\frac{\mr{d}\bar{V}_\mf{x}}{\mr{d}t} &= \bar{Q}_\mf{x}+\mf{K}_{\mf{w}_\mr{e}}^\mr{T}\bar{Q}_{\mf{w}_\mr{e}}+\mf{K}_{\mf{w}_\mr{p}}^\mr{T}\bar{Q}_{\mf{w}_\mr{p}} + \bar{Q}_{\mf{w}_\mr{e}\mf{x}}^\mr{T}\ell_{\mf{w}_\mr{e}}+ \bar{Q}_{\mf{w}_\mr{p}\mf{x}}^\mr{T}\ell_{\mf{w}_\mr{p}} + \mf{K}_{\mf{w}_\mr{e}}^\mr{T}\bar{Q}_{\mf{w}_\mr{e}\mf{w}_\mr{e}}\ell_{\mf{w}_\mr{e}}+ \mf{K}_{\mf{w}_\mr{p}}^\mr{T}\bar{Q}_{\mf{w}_\mr{p}\mf{w}_\mr{p}}\ell_{\mf{w}_\mr{p}}, \\
-\frac{\mr{d}\bar{V}_{\mf{x}\mf{x}}}{\mr{d}t} &=\mf{K}_{\mf{w}_\mr{e}}^\mr{T}\bar{Q}_{\mf{w}_\mr{e}\mf{x}} + \bar{Q}_{\mf{w}_\mr{e}\mf{x}}^\mr{T}\mf{K}_{\mf{w}_\mr{e}}
+\mf{K}_{\mf{w}_\mr{p}}^\mr{T}\bar{Q}_{\mf{w}_\mr{p}\mf{x}} + \bar{Q}_{\mf{w}_\mr{p}\mf{x}}^\mr{T}\mf{K}_{\mf{w}_\mr{p}}
+ \mf{K}_{\mf{w}_\mr{e}}^\mr{T}\bar{Q}_{\mf{w}_\mr{e}\mf{w}_\mr{e}} \mf{K}_{\mf{w}_\mr{e}}+ \mf{K}_{\mf{w}_\mr{p}}^\mr{T}\bar{Q}_{\mf{w}_\mr{p}\mf{w}_\mr{p}} \mf{K}_{\mf{w}_\mr{p}} + \bar{Q}_{\mf{x}\mf{x}},
\end{split}
\end{equation}
where, under a second-order approximation of the boundary condition of \eqref{eq:HJI}, we have
\begin{equation}\label{eq:boundary}
\bar{V}(t_f)=\phi(\bar{\mf{x}}(t_f), t_f),\quad
\bar{V}_{\mf{x}}(t_f)=\phi_{\mf{x}}(\bar{\mf{x}}(t_f), t_f),\quad
\bar{V}_{\mf{x}\mf{x}}(t_f)=\phi_{\mf{x}\mf{x}}(\bar{\mf{x}}(t_f), t_f).
\end{equation}
DDP then iteratively solves the forward perturbation equation \eqref{eq:forward} and the backward equations \eqref{eq:backward} until convergence. The full procedure is described in Algorithm \ref{al:DDP}.

\begin{algorithm}
\caption{Game-Theoretic DDP for Adversarial Pursuits in Cislunar Space}
\hspace*{\algorithmicindent} \textbf{Input}: Initial condition $\mf{x}_0$, initial evader and pursuer policies $\{\bar{\mf{w}}_\mr{e},~\bar{\mf{w}}_\mr{p}\}$, initial and final times $t_0,~t_f$, tolerance $\epsilon>0$.\\
\hspace*{\algorithmicindent} \textbf{Output}: Evader and pursuer policies $\mathbf{w}_\mr{e}^\star$, $\mathbf{w}_\mr{p}^\star$, feedforward gains $\ell_{\mf{w}_\mr{e}}, \ell_{\mf{w}_\mr{p}}$, feedback gains $\mf{K}_{\mf{w}_\mr{e}}, \mf{K}_{\mf{w}_\mr{p}}$.
\begin{algorithmic}[1]
\Procedure{}{}
\While{$\|\delta\mathbf{w}_\mr{e}^\star\| + \|\delta\mathbf{w}_\mr{p}^\star\| > \epsilon$}
    \State \textbf{Forward rollout (nonlinear):} Propagate $\bar{\mathbf{x}}$ over $[t_0,t_f]$ under $(\bar{\mathbf{w}}_\mr{e},\bar{\mathbf{w}}_\mr{p})$ from \eqref{eq:nonlineardyn}.
    \State \textbf{Terminal conditions:} Compute $(\bar{V}(t_f),\bar{V}_{\mathbf{x}}(t_f),\bar{V}_{\mathbf{x}\mathbf{x}}(t_f))$ from \eqref{eq:boundary}.
    \State \textbf{Backward sweep:} Propagate $(\bar{V},\bar{V}_{\mathbf{x}},\bar{V}_{\mathbf{x}\mathbf{x}})$ over $[t_0,~t_f]$ from \eqref{eq:backward}.
    \State \textbf{Gain Computation:} Compute $\ell_{\mf{w}_\mr{e}}, \ell_{\mf{w}_\mr{p}}$ and $\mf{K}_{\mf{w}_\mr{e}}, \mf{K}_{\mf{w}_\mr{p}}$ from \eqref{eq:gains}.
    \State \textbf{Forward rollout (linearized):} Propagate $\delta\bar{\mathbf{x}}$ over $[t_0,t_f]$ from \eqref{eq:forward} under $\delta{\mf{w}_\mr{e}}=\delta\mf{w}_\mr{e}^\star$, $\delta{\mf{w}_\mr{p}}=\delta\mf{w}_\mr{p}^\star$ in \eqref{eq:updates}.
    \State \textbf{Control refinement:} Update nominal controllers
    \[
      \bar{\mathbf{w}}_\mr{e} \leftarrow \bar{\mathbf{w}}_\mr{e} + \gamma\,\delta\mathbf{w}_\mr{e}^\star,\qquad
      \bar{\mathbf{w}}_\mr{p} \leftarrow \bar{\mathbf{w}}_\mr{p} + \gamma\,\delta\mathbf{w}_\mr{p}^\star,
    \]
    \qquad\quad where $\gamma\in(0,1]$.
    \State \textbf{Control update:} Set $\mathbf{w}_\mr{e}^\star=\bar{\mathbf{w}}_\mr{e}$ and $\mathbf{w}_\mr{p}^\star=\bar{\mathbf{w}}_\mr{p}$.
\EndWhile
\EndProcedure
\end{algorithmic}\label{al:DDP}
\end{algorithm}

\begin{remark}
Since the adversarial pursuit takes place about a nominal cislunar orbit that is an input-free solution of the CR3BP dynamics, a natural initialization of Algorithm \ref{al:DDP} is $\bar{\mf{u}}_\mr{e}=\bar{\mf{u}}_\mr{p}=0$ and $\tau_\mr{e}=\tau_\mr{p}=1$. Moreover, for periodic orbits with period $T$, it is natural to select $t_f=t_0+T$ so that phases of the cislunar orbit with enhanced controllability are captured within the prediction horizon. For instance, the near-rectilinear halo orbit in Figure \ref{fig:halo_space} generally exhibits greater controllability near perilune. 
\end{remark}

\subsection{Regularization of Continuous-Time DDP}

In discrete-time implementations of DDP, it is common for the Hessian 
$\bar{\mr{Q}}_{\mf{w}_\mr{e}\mf{w}_\mr{e}}$, and by extension the Hessian 
$\bar{\mr{Q}}_{\mf{w}_\mr{p}\mf{w}_\mr{p}}$, to become indefinite during the backward 
pass. This loss of definiteness can cause divergence of the algorithm and therefore 
motivates the use of regularization to enforce strict definiteness \cite{regularization}. In contrast, 
in the continuous-time DDP setting we consider, both 
$\bar{\mr{Q}}_{\mf{w}_\mr{e}\mf{w}_\mr{e}}$ and $\bar{\mr{Q}}_{\mf{w}_\mr{p}\mf{w}_\mr{p}}$ are 
guaranteed to be strictly definite, since we have the formulas
\begin{equation*}
\bar{\mr{Q}}_{\mf{w}_\mr{e}\mf{w}_\mr{e}} = 
\begin{bmatrix}2R_\mr{e} & 0 \\ 0 & 2a_\mr{e}\end{bmatrix}\succ0, 
\qquad 
\bar{\mr{Q}}_{\mf{w}_\mr{p}\mf{w}_\mr{p}} = 
\begin{bmatrix}-2R_\mr{p} & 0 \\ 0 & -2a_\mr{p}\end{bmatrix}\prec 0.
\end{equation*}
However, unlike in the discrete-time case, strict definiteness of these matrices 
does not by itself ensure well-posedness of the backward pass. This is because the 
associated differential Riccati equation for $\bar{V}_{\mf{x}\mf{x}}$ in~\eqref{eq:backward} is quadratic and may  exhibit finite-time blow-up even when 
$\bar{\mr{Q}}_{\mf{w}_\mr{e}\mf{w}_\mr{e}}\succ0$ and 
$\bar{\mr{Q}}_{\mf{w}_\mr{p}\mf{w}_\mr{p}}\prec0$ hold uniformly. 
For this reason, we argue that regularization of 
$\bar{\mr{Q}}_{\mf{w}_\mr{e}\mf{w}_\mr{e}}$ and $\bar{\mr{Q}}_{\mf{w}_\mr{p}\mf{w}_\mr{p}}$ 
remains necessary during the backward pass of DDP, and this regularization should be chosen 
large enough to guarantee the existence of unique, bounded solutions to \eqref{eq:backward}. 
The following result formalizes this guarantee.

\begin{theorem}
Consider the backward differential equations \eqref{eq:backward}, with $\bar{\mr{Q}}_{\mf{w}_\mr{e}\mf{w}_\mr{e}}$ and 
$\bar{\mr{Q}}_{\mf{w}_\mr{p}\mf{w}_\mr{p}}$ substituted with $\bar{\mr{Q}}_{\mf{w}_\mr{e}\mf{w}_\mr{e}}+\lambda I_{4}$ and 
$\bar{\mr{Q}}_{\mf{w}_\mr{p}\mf{w}_\mr{p}}-\lambda I_{4}$, $\lambda\ge0$. Let $\bar{\mf{x}}$ be a continuous solution of \eqref{eq:nonlineardyn}. Then, there exists $\lambda^\star\ge0$ such that if $\lambda\ge\lambda^\star$ then the equations \eqref{eq:backward} admit unique, bounded solutions over $t\in[t_0,~t_f]$.
\end{theorem}
\begin{proof}
The right-hand sides of~\eqref{eq:backward} are continuous on 
$t$ and continuously differentiable in $(\bar{V},\bar{V}_\mf{x},\bar{V}_\mf{xx})$, hence admit unique solutions \cite{khalil2002nonlinear}.  Moreover, if we replace $\bar{\mr{Q}}_{\mf{w}_\mr{e}\mf{w}_\mr{e}}$ and 
$\bar{\mr{Q}}_{\mf{w}_\mr{p}\mf{w}_\mr{p}}$ with 
$\bar{\mr{Q}}_{\mf{w}_\mr{e}\mf{w}_\mr{e}}+\lambda I_{4}$ and 
$\bar{\mr{Q}}_{\mf{w}_\mr{p}\mf{w}_\mr{p}}-\lambda I_{4}$, respectively, then as 
$\lambda\to\infty$, the right-hand sides of~\eqref{eq:backward} converge uniformly
on any compact sets of $(t,\bar{V},\bar{V}_\mf{x},\bar{V}_\mf{xx})$ to
\begin{align*}
-\frac{\mr{d}\bar{V}}{\mr{d}t} &= \bar{L}+\ell_{\mf{w}_\mr{e}}^\textrm{T}\bar{Q}_{\mf{w}_\mr{e}}+\ell_{\mf{w}_\mr{p}}^\textrm{T}\bar{Q}_{\mf{w}_\mr{p}} + \frac{1}{2}\ell_{\mf{w}_\mr{e}}^\textrm{T}\bar{Q}_{\mf{w}_\mr{e}\mf{w}_\mr{e}}\ell_{\mf{w}_\mr{e}} + \frac{1}{2}\ell_{\mf{w}_\mr{p}}^\textrm{T}\bar{Q}_{\mf{w}_\mr{p}\mf{w}_\mr{p}}\ell_{\mf{w}_\mr{p}}\xrightarrow{\lambda \to \infty} \bar{L}, \\
-\frac{\mr{d}\bar{V}_\mf{x}}{\mr{d}t} &= \bar{Q}_\mf{x}+\mf{K}_{\mf{w}_\mr{e}}^\mr{T}\bar{Q}_{\mf{w}_\mr{e}}+\mf{K}_{\mf{w}_\mr{p}}^\mr{T}\bar{Q}_{\mf{w}_\mr{p}} + \bar{Q}_{\mf{w}_\mr{e}\mf{x}}^\mr{T}\ell_{\mf{w}_\mr{e}}+ \bar{Q}_{\mf{w}_\mr{p}\mf{x}}^\mr{T}\ell_{\mf{w}_\mr{p}} + \mf{K}_{\mf{w}_\mr{e}}^\mr{T}\bar{Q}_{\mf{w}_\mr{e}\mf{w}_\mr{e}}\ell_{\mf{w}_\mr{e}}+ \mf{K}_{\mf{w}_\mr{p}}^\mr{T}\bar{Q}_{\mf{w}_\mr{p}\mf{w}_\mr{p}}\ell_{\mf{w}_\mr{p}}\xrightarrow{\lambda \to \infty}  \bar{F}_{\mf{x}}^\mr{T}\bar{V}_{\mf{x}} + \bar{L}_{\mf{x}}, \\
-\frac{\mr{d}\bar{V}_{\mf{x}\mf{x}}}{\mr{d}t} &=\mf{K}_{\mf{w}_\mr{e}}^\mr{T}\bar{Q}_{\mf{w}_\mr{e}\mf{x}} + \bar{Q}_{\mf{w}_\mr{e}\mf{x}}^\mr{T}\mf{K}_{\mf{w}_\mr{e}}
+\mf{K}_{\mf{w}_\mr{p}}^\mr{T}\bar{Q}_{\mf{w}_\mr{p}\mf{x}} + \bar{Q}_{\mf{w}_\mr{p}\mf{x}}^\mr{T}\mf{K}_{\mf{w}_\mr{p}}
+ \mf{K}_{\mf{w}_\mr{e}}^\mr{T}\bar{Q}_{\mf{w}_\mr{e}\mf{w}_\mr{e}} \mf{K}_{\mf{w}_\mr{e}}\\&\qquad\qquad\qquad\qquad\qquad\qquad\qquad\qquad\qquad\qquad\qquad + \mf{K}_{\mf{w}_\mr{p}}^\mr{T}\bar{Q}_{\mf{w}_\mr{p}\mf{w}_\mr{p}} \mf{K}_{\mf{w}_\mr{p}} + \bar{Q}_{\mf{x}\mf{x}}\xrightarrow{\lambda \to \infty}\bar{L}_{\mf{x}\mf{x}} + \bar{V}_{\mf{x}\mf{x}}\bar{F}_{\mf{x}} 
+ \bar{F}_{\mf{x}}^\mr{T}\bar{V}_{\mf{x}\mf{x}}.
\end{align*}
This limiting system is linear in 
$(\bar{V},\bar{V}_\mf{x},\bar{V}_\mf{xx})$ and continuous in $t$, and therefore admits bounded 
solutions over $t\in[t_0,\,t_f]$ \cite{khalil2002nonlinear}. By continuity with respect to~$\lambda$, there exists 
$\lambda^\star\ge 0$ such that for all $\lambda\ge\lambda^\star$, the original 
system~\eqref{eq:backward} also has bounded solutions over $t\in[t_0,t_f]$, concluding the proof. \frQED
\end{proof}

Following this, at each backward pass of DDP, we add a sufficiently large regularization to $\bar{\mr{Q}}_{\mf{w}_\mr{e}\mf{w}_\mr{e}}$ and $\bar{\mr{Q}}_{\mf{w}_\mr{p}\mf{w}_\mr{p}}$ until the backward pass \eqref{eq:backward} becomes well-defined.

\subsection{Cost Function Shaping using Stable and Unstable Manifolds}

DDP is generally guaranteed to converge only to a local saddle point of the game \eqref{eq:nlgame}. The quality of this solution can depend on the initial evader-pursuer policies but also on the weighting parameters of the game. For this reason, we design the cost matrices $Q_{\mr{e}},~Q_{\mr{p}}$ to guide DDP toward regions that are less risky and where less fuel is required to track the reference orbit, while steering it away from regions with higher fuel demands. These correspond, respectively, to the orbit’s \textit{stable} and \textit{unstable} manifolds.

In the CR3BP, periodic orbits around Lagrange points have associated stable and unstable manifolds, which form tubes of trajectories in phase space: the stable manifold is the set of trajectories that asymptotically converge back to the orbit as time advances, while the unstable manifold comprises trajectories that diverge away from it. These structures arise from the eigen-directions of the linearized dynamics, and their effect depends on the stability of the orbit; for example, the drift along the unstable manifold of a halo orbit is typically slower than that of a Lyapunov orbit.

To compute the stable and the unstable manifolds of a periodic orbit $\mathbf{x}_{\mr{d}}$, we follow the method described in \cite{koon2000dynamical}. Specifically, we compute the monodromy matrix $\Phi(t,t_0)$ of the periodic orbit by integrating the differential equation
\begin{equation*}
\dot{\Phi}(t,t_0)=A(\mathbf{x}_d(t))\Phi(t,t_0),~\Phi(t_0)=I_6,
\end{equation*}
 where $A(\mathbf{x}_d(t))$ is the Jacobian of the system~\eqref{eq:CR3BP} evaluated at $\mathbf{x}_d(t)$ (see Appendix). After one period $T$, the eigenvalues and eigenvectors of the monodromy matrix, i.e., of $\Phi(t_0+T,t_0)$,
 are evaluated to characterize the orbit’s stability properties. Eigenvalues with magnitude greater than unity correspond to exponentially growing directions (the unstable manifold), whereas those with magnitude less than unity correspond to exponentially decaying directions (the stable manifold). The associated eigenvectors at the initial state $\mathbf{x}_\mr{d}(t_0)$, denoted $\mf{e}_{u0}$ and $\mf{e}_{s0}$, define the local unstable and stable directions, respectively. Their time evolution is governed by the equations 
 \begin{align}
    \mf{e}_{s}(t)=\frac{\Phi(t,t_0)\mf{e}_{s0}}{\norm{\Phi(t,t_0)\mf{e}_{s0}}}, \hspace{2cm }\mf{e}_{u}(t)=\frac{\Phi(t,t_0)\mf{e}_{u0}}{\norm{\Phi(t,t_0)\mf{e}_{u0}}}.
\end{align}

Because perturbations along the unstable manifold dominate long-term divergence, we embed these directions explicitly in the cost function, ensuring that the pursuit-evasion strategies account for the most destabilizing modes of the orbit. To this end, we introduce the projector
\begin{equation*}
P_u(t) = \mf{e}_{u}(t)\mf{e}_{u}^\mr{T}(t)
\end{equation*}
 to isolate the component of the state along the unstable manifold. We then define the reference tracking penalty matrices as 
\begin{equation}\label{eq:manifolds}
\begin{split}
    Q_\mr{i}(t) = \alpha Q_\mr{i0} + (1-\alpha) Q_\mr{i0}^{1/2}P_u(c_\mr{i}(t))Q_\mr{i0}^{1/2}, \qquad\mr{i}\in\{\mr{e},\mr{p}\},\\ F_\mr{i} = \alpha F_\mr{i0} + (1-\alpha) F_\mr{i0}^{1/2}P_u(c_\mr{i}(t_f))F_\mr{i0}^{1/2}, \qquad \mr{i}\in\{\mr{e},\mr{p}\},
\end{split}
\end{equation}
where $Q_\mr{i0},F_\mr{i0}\succ0$ denotes the baseline weighting on position and velocity states, and $\alpha \in [0,1]$ provides a tunable parameter to balance between uniform penalties and enhanced penalization of deviations in the unstable direction. This construction guides DDP toward evading and pursuing policies that remain away from unstable manifolds and thereby require less fuel and are less risky. Note that $Q_\mr{i0},~F_\mr{i0}$ must be block diagonals for the unstable manifold direction to remain undistorted.

\begin{figure}[!t]
    \centering
    \colorbox{black}{\makebox[\dimexpr\linewidth][c]{%
        \includegraphics[width=1\linewidth]{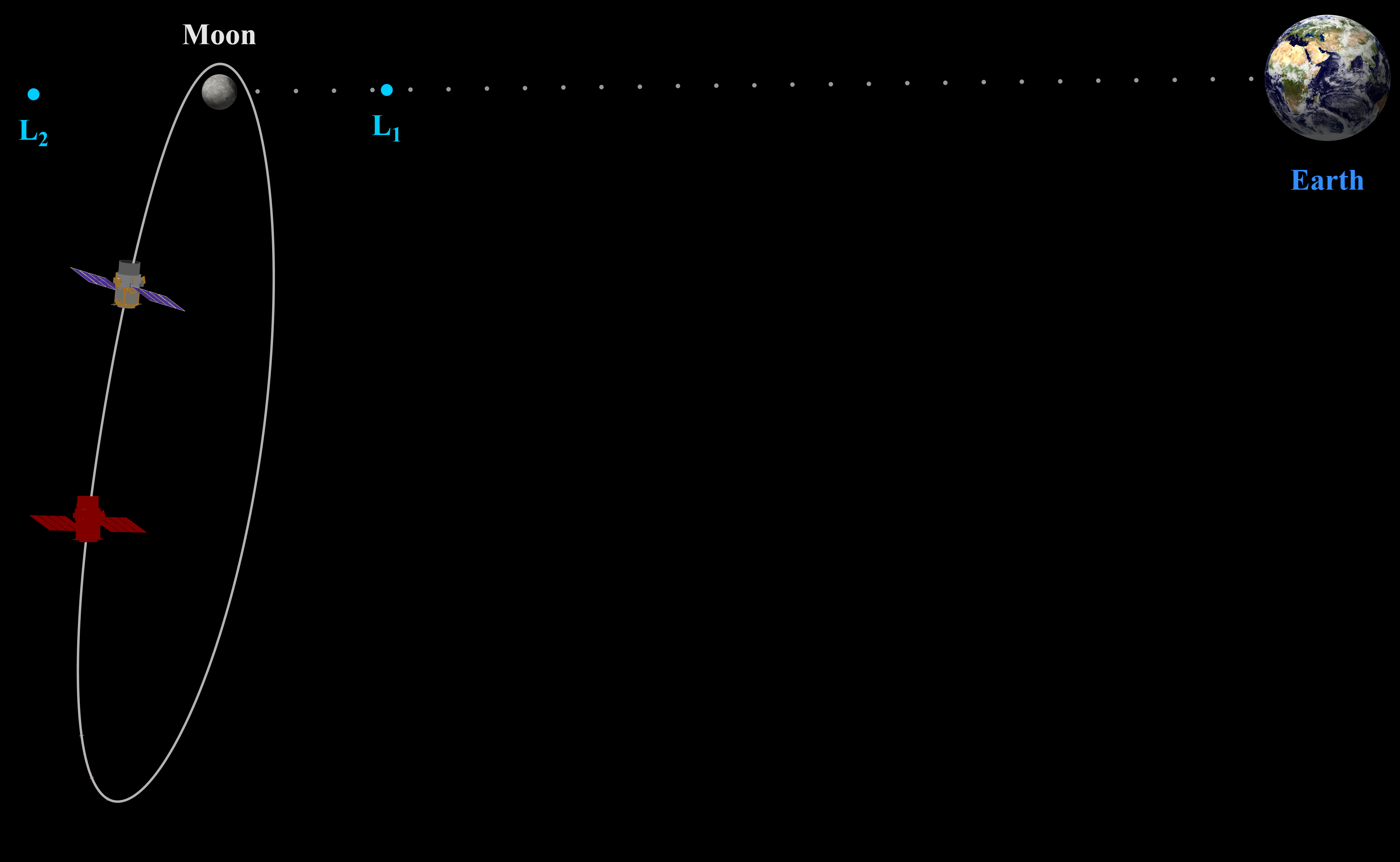}%
    }}
    \caption{The simulated near-rectilinear halo orbit. }
    \label{fig:halo_space}
\end{figure}

\section{Numerical Experiments}

In this section, we perform numerical experiments to illustrate the pursuit-evasion policies. Throughout these, to express the thrusts $\mf{u}_\mr{e},~\mf{u}_\mr{p}$ of the spacecraft directly in Newtons, we scale the input matrices $B_\mr{e}$ and $B_\mr{p}$ by the factor $\mr{TU}^2/\mr{LU}$, where $\mr{TU}$ and $\mr{LU}$ denote the characteristic time and length scales used in the non-dimensionalization of the CR3BP equations.

\subsection{Simulation of Saddle-Point Pursuit-Evasion Policies}

We consider an evader-pursuer interaction along the near-rectilinear halo orbit illustrated in Figure~\ref{fig:halo_space}, which has a period of $T=6.5~\textrm{days}$ (i.e., $1.466695~[\textrm{ND}]$). The evader’s objective is to drive its separation from the pursuer beyond $600~\textrm{km}$, while the pursuer aims to minimize this distance. Both spacecraft have an equal mass of $m_\mathrm{e}=m_\mathrm{p}=1000~\mathrm{kg}$, and are constrained to remain in the vicinity of the relatively stable reference halo orbit. For the initial configuration, the pursuer is at apolune and the evader is $6.38$ minutes ahead along the orbit (corresponding to $0.001$ in non-dimensional time).

We generate the control actions of the spacecraft according to the policies that solve the nonlinear differential game~\eqref{eq:nlgame}. We apply these policies in a model predictive control fashion, with a prediction horizon of $t_f=T$, a control horizon of $T/5$ (i.e., five control updates per orbital period), and $t_0=0$. We choose the cost parameters of the game~\eqref{eq:nlgame} as
\begin{align*}
&R_\mr{e}=0.025I_3, \quad R_\mr{p}=0.05I_3, \quad a_\mr{e}=0.005, \quad a_\mr{p}=0.01, \\
&Q_\mr{e0}=F_\mr{e0}=Q_\mr{p0}=F_\mr{p0}=5I_6, \quad w=2000, \quad p=2.1,
\end{align*}
which implies the evader spacecraft is twice as fast as the pursuer spacecraft.
We further select $d_0=660~\textrm{km}$, which is $10\%$ larger than the actual separation objective to account for the vanishing property of~\eqref{eq:evadecost} near $d_0$. To enable more aggressive evading maneuvers, we design the reference-tracking cost matrices according to~\eqref{eq:manifolds}. Specifically, the evader increases its aggressiveness when close to the pursuer by setting
\[
\alpha\big|_{t=t_r} = \max\left\{1 - \frac{\min_{t\in[t_r-T,t_r]}\norm{\mf{p}_\mr{e}(t)-\mf{p}_\mr{p}(t)}}{d_0},~0\right\},
\]
where $t_r$ is the replanning time instant. In other words, the evader monitors the minimum separation from the pursuer over the most recent orbital period and adjusts its evasion cost accordingly.

\begin{figure}[!t]
    \centering
        \includegraphics[width=0.94\linewidth]{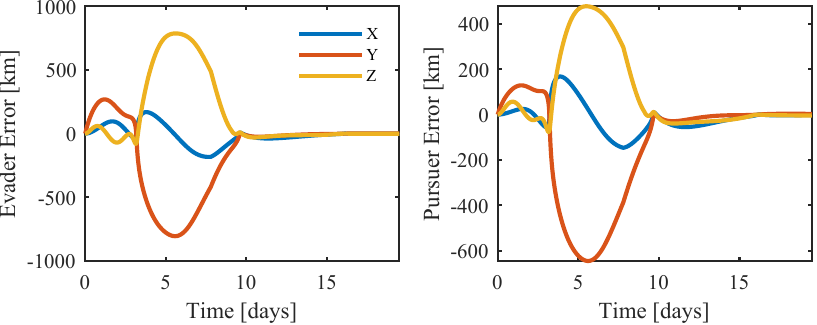}%
    \caption{The evolution of the position tracking errors $\mf{p}_\mr{i}-\mf{p}_{\mr{di}}$, $\mr{i}\in\{\mr{e},\mr{p}\}$, of the evader and the pursuer. }
    \label{fig:track}
\end{figure}

\begin{figure}[!t]
    \centering
        \includegraphics[width=0.94\linewidth]{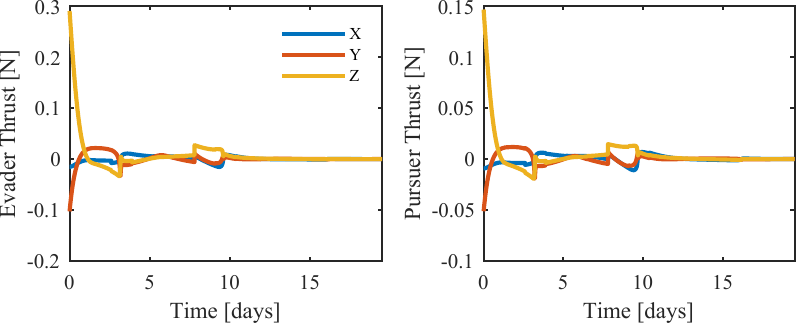}%
    \caption{The evolution of the thrust profiles $\mf{u}_\mr{i}$, $\mr{i}\in\{\mr{e},\mr{p}\}$, of the evader and the pursuer.}
    \label{fig:thrust}
\end{figure}

\begin{figure}[!t]
    \centering
        \includegraphics[width=0.94\linewidth]{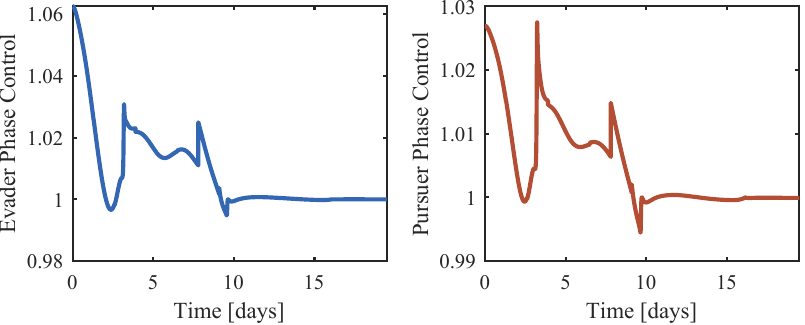}%
    \caption{The evolution of the phasing controls $\tau_\mr{i}$, $\mr{i}\in\{\mr{e},\mr{p}\}$, of the evader and the pursuer.}
    \label{fig:speedup}
\end{figure}

To solve the differential game~\eqref{eq:nlgame}, we apply the DDP Algorithm~\ref{al:DDP} initialized with $\bar{\mf{w}}_\mr{e}=\bar{\mf{w}}_\mr{p}=0$ and a tolerance of $\epsilon=10^{-4.5}$. We solve the forward equations with the MATLAB function \texttt{ode45}, and the backward equations with \texttt{ode113}. Whenever the backward integration fails, we regularize the equations by multiplying $\bar{\mr{Q}}_{\mf{w}_\mr{e}\mf{w}_\mr{e}}, \bar{\mr{Q}}_{\mf{w}_\mr{p}\mf{w}_\mr{p}}$ by $1.5$; whenever it succeeds, we de-regularize by dividing them by $1.5$ until their nominal values are restored. We select the parameter $\gamma$ from the set $\{1, 0.5, 0.25, 0.1\}$, choosing the first value that decreases the evader’s cost. Finally, to improve convergence, we initialize each DDP call with the optimal control sequences from the previous call, shifted forward in time to match the spacecraft’s current position along the orbit.

Figures \ref{fig:track}-\ref{fig:speedup} show the tracking errors $\mf{p}_\mr{i}-\mf{p}_{\mr{di}}=[I_3 ~ 0_3](\mf{x}_\mr{i}-\mf{x}_{\mr{di}})$, $\mr{i}\in\{\mr{e},\mr{p}\}$, of the pursuer and the evader, along with their thrusts and phasing controls. We observe that, during the first two periods, the tracking errors increase as the two spacecraft engage in a pursuit; however, the increase is modest relative to the scale of the halo orbit, because the spacecraft are constrained by their reference tracking objectives. Moreover, we observe that both spacecraft increase the phasing on their orbit as a means to evade/pursue, with the evader being faster by assumption.

\begin{figure}[!t]
    \centering
        \includegraphics[width=1\linewidth]{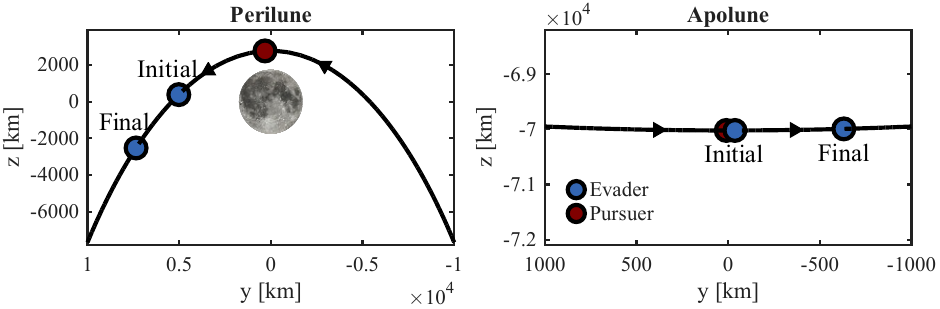}%
    \caption{The reference position of the evader at the instants the pursuer reaches perilune (left) and apolune (right), during the first and last revolutions about the near-rectilinear halo orbit.}
    \label{fig:zoom}
\end{figure}

\begin{figure}[!t]
    \centering
        \includegraphics[width=0.55\linewidth]{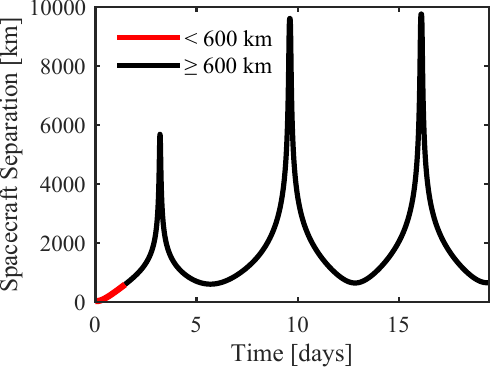}%
    \caption{The evolution of the separation $\norm{\mf{p}_\mr{e}-\mf{p}_\mr{p}}$ between the evader and the pursuer.}
    \label{fig:separation}
\end{figure}

Figure \ref{fig:zoom} shows the position of the evader during the initial and the final orbital period of the simulation, at the time instant when the pursuer is at perilune and apolune, respectively. Figure \ref{fig:separation} also shows the separation between the two spacecraft. We notice that the separation is generally more pronounced at perilune and less pronounced at apolune. This is because the trajectories evolve much faster at the former, and slow down significantly at the latter. In addition, we observe that the evader was able to permanently cross the $600~\textrm{km}$ separation threshold after only $1.5$ days, despite the use of low-thrust engines and the constraint of staying in the vicinity of the halo orbit. This showcases the effectiveness of the evader’s strategy in leveraging orbital dynamics to achieve sustained separation, even under strict thrust capabilities and orbital constraints, and highlights the potential for low-thrust spacecraft to perform meaningful evasive maneuvers in cislunar environments.

\subsection{Effect of Reference Phasing Control and Cost-Function Shaping}
To isolate the roles of reference phasing control and manifold-based cost shaping, we reran the study in three ablations: (i) cost-function shaping removed, (ii) phasing control turned off, and (iii) both mechanisms removed.
Figure \ref{fig:comp_separation_all} illustrates the effect of disabling phasing control and cost shaping. 

Reference phasing control has the most dominant impact: by adjusting the speed of its reference trajectory, the evader can accelerate or decelerate along-track to break phase lock with the pursuer. This prevents repeated close approaches while keeping thrust costs low, since control acts primarily in the tangential direction.
On the other hand, cost shaping with manifold information, in the absence of phase control, does provide some additional flexibility, but it is not enough on its own. In this case, the evader can avoid capture only by drifting far from its reference, resulting in inefficient thrusting. Moreover, when both mechanisms are disabled, the separation repeatedly falls well below the 600 km threshold, confirming that phasing control is indispensable and that cost shaping provides additional control authority. 

Finally, we conclude that combining both reference phasing control and cost shaping is the most effective approach for designing evading maneuvers. On the one hand, phasing control is essential, as Figure \ref{fig:comp_separation_all} illustrates. On the other hand, manifold shaping enhances the authority of phase control, enabling a quicker increase in spacecraft separation (Figure \ref{fig:separation}).
 Without it, comparable evading performance required slowing down the pursuer by a factor of 2.5 (i.e., increasing cost parameters to $a_\mr{e}=0.05$, $a_\mr{p}=0.25$).

\begin{figure}[!t]
  \centering
  \includegraphics[width=\linewidth]{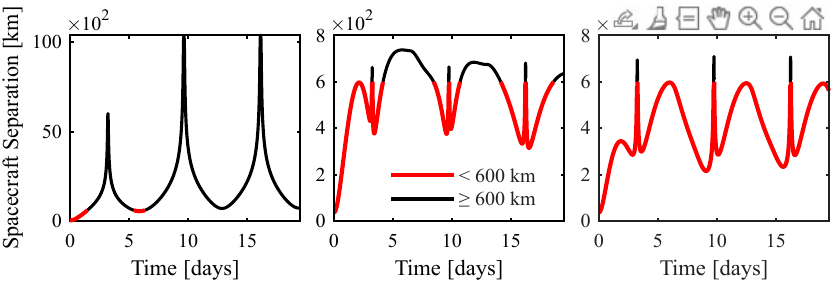}
  \caption{Evolution of the separation $\norm{\mf{p}_\mr{e}-\mf{p}_\mr{p}}$ under three ablations: (left) cost-function shaping removed, (middle) phasing control turned off, and (right) both mechanisms removed. }
  \label{fig:comp_separation_all}
\end{figure}

\begin{figure}[!ht]
    \centering
        \includegraphics[width=0.55\linewidth]{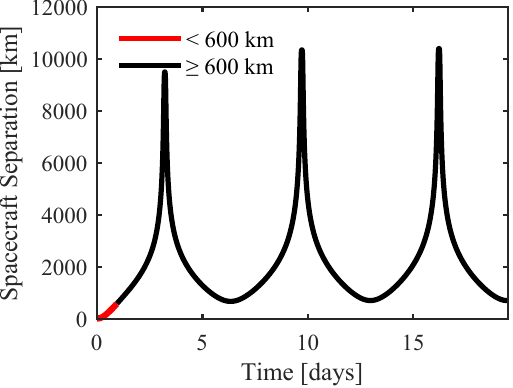}%
    \caption{The evolution of the separation $\norm{\mf{p}_\mr{e}-\mf{p}_\mr{p}}$ between the evader and the pursuer, when the pursuer uses a linear-quadratic policy.}
    \label{fig:separation2}
\end{figure}

\subsection{Evasion Against a Linear-Quadratic Baseline Pursuer}

To showcase the effectiveness of the evading policy against pursuing policies different from the saddle-point equilibrium, we simulate a pursuer whose policy is obtained from a linear-quadratic pursuit-evasion game. We describe this game and its solution in the Appendix. We select the parameters of the game to match those of the nonlinear dynamic game setup, with the weighting matrices of the evasion objective set to $M = M_f = 100~I_3$. When the pursuer-evader distance exceeds $600$~km, we assume the pursuer ceases pursuit and reverts to a standard linear-quadratic tracking controller (i.e., $M=M_f=0$).

Figure \ref{fig:separation2} shows the resulting separation between the two spacecraft. We observe that the evader crosses the $600$ km threshold in less than a day, which is faster than in the case where the pursuer uses the saddle-point policy of the nonlinear dynamic game. This distance also remains larger throughout the simulation. This highlights that, without appropriate modifications, simple linear-quadratic games are not suitable for adversarial pursuits in cislunar space.

\section{Conclusion}

We study adversarial pursuits in cislunar space by formulating and solving a nonlinear differential pursuit-evasion game in the CR3BP. This game incorporates critical aspects of the orbital geometry of cislunar space, including orbital phasing control and costs informed by stable/unstable manifolds. Simulations demonstrate that the combination of these components enables the evader to quickly and permanently increase its distance from the pursuer.

Future work will consider modeling of the potential jamming between the evader and the pursuer, following \cite{fotiadis2025optimal}.

\section*{Acknowledgments}
\vspace{3mm}
This work was sponsored in part by grants AFRL FA9550-23-1-0646 and AFOSR FA9550-22-1-0403.

\bibliography{sample}

\section*{Appendix: Linear-Quadratic Game for Cislunar Adversarial Pursuits}

Here, we formulate the adversarial pursuit in cislunar space as a linear-quadratic game, which requires a first-order approximation of the spacecraft dynamics \eqref{eq:spacecraftdyn} around the nominal orbits \eqref{eq:refs}. In that regard, we restrict attention to the case where the two spacecraft share the same nominal orbit, i.e., $c_\mr{p}(t)=c_\mr{e}(t)$ and hence $\mathbf{x}_\mr{de}(t)=\mathbf{x}_\mr{dp}(t)=\mathbf{x}_\mr{d}(t)$.

\subsection{Linear-Quadratic Game Formulation}

Define the orbital tracking errors as $\Delta\mathbf{x}_\textrm{i}=\mathbf{x}_\textrm{i}-\mathbf{x}_\mathrm{d}$ for $\mr{i}\in\{\mr{e}, \mr{p}\}$. Then, standard linear-systems theory yields the local approximation $\Delta\dot{\mathbf{x}}_\textrm{i}= A(\mathbf{x}_\mr{d}(t))\Delta\mathbf{x}_\textrm{i}+B_\textrm{i}\mf{u}_\textrm{i}$, $\mr{i}\in\{\mathrm{e}, \textrm{p}\}$, where $A(\mathbf{x}_\mr{d}(t))=\frac{\partial f(\mathbf{x})}{\partial \mathbf{x}}\big|_{\mathbf{x}=\mathbf{x}_\mathrm{d}(t)}$. Since $\mathbf{x}_\mr{d}$ is a nominal cislunar orbit that solves \eqref{eq:CR3BPc} under zero control input, the linearization is naturally taken about zero input. Moreover, using \eqref{eq:CR3BP}, we can calculate the Jacobian of the three-body dynamics function $f$ as
\begin{align}\nonumber
\frac{\partial f(\mathbf{x})}{\partial \mathbf{x}}=\begin{bmatrix}0 & 0 & 0 & 1 & 0 & 0 \\ 0 & 0 & 0 & 0 & 1 & 0 \\ 0 & 0 & 0 & 0 & 0 & 1 \\ f_{\dot{x}{x}}(\mathbf{x}) & f_{\dot{x}{y}}(\mathbf{x}) & f_{\dot{x}{z}}(\mathbf{x}) & 0 & 2 & 0\\ f_{\dot{y}{x}}(\mathbf{x}) & f_{\dot{y}{y}}(\mathbf{x}) & f_{\dot{y}{z}}(\mathbf{x}) & -2 & 0 & 0 \\ f_{\dot{z}{x}}(\mathbf{x}) & f_{\dot{z}{y}}(\mathbf{x}) & f_{\dot{z}{z}}(\mathbf{x})  &  0 & 0 & 0\end{bmatrix}
\end{align}
with
\begin{align*}
f_{\dot{x}{x}}(\mathbf{x})&=1-\frac{1-\mu}{r_e^3}+\frac{3(1-\mu)(x+\mu)^2}{r_e^5}-\frac{\mu}{r_m^3}+\frac{3\mu(x-1+\mu)^2}{r_m^5},\\
f_{\dot{x}{y}}(\mathbf{x})&=\frac{3(1-\mu)(x+\mu)y}{r_e^5}+\frac{3\mu(x-1+\mu)y}{r_m^5},\\
f_{\dot{x}{z}}(\mathbf{x})&=\frac{3(1-\mu)(x+\mu)z}{r_e^5}+\frac{3\mu(x-1+\mu)z}{r_m^5},\\
f_{\dot{y}{y}}(\mathbf{x})&=1-\frac{1-\mu}{r_e^3}+\frac{3(1-\mu)y^2}{r_e^5}-\frac{\mu}{r_m^3}+\frac{3\mu y^2}{r_m^5},\\
f_{\dot{y}{z}}(\mathbf{x})&=\frac{3(1-\mu)yz}{r_e^5}+\frac{3\mu yz}{r_m^5},\\
f_{\dot{z}{z}}(\mathbf{x})&=-\frac{1-\mu}{r_e^3}+\frac{3(1-\mu)z^2}{r_e^5}-\frac{\mu}{r_m^3}+\frac{3\mu z^2}{r_m^5},\\
f_{\dot{x}{y}}(\mathbf{x})&=f_{\dot{y}{x}}(\mathbf{x}),~f_{\dot{x}{z}}(\mathbf{x})=f_{\dot{z}{x}}(\mathbf{x}),~f_{\dot{y}{z}}(\mathbf{x})=f_{\dot{z}{y}}(\mathbf{x}).
\end{align*}

 Subsequently, note that the objective of the evader (pursuer) is to maximize (minimize) the evader-pursuer distance while tracking its nominal reference \eqref{eq:refs}. Hence, defining the concatenated state $\Delta\mathbf{x}=[\Delta\mathbf{x}_\mr{e}^\mr{T}~\Delta\mathbf{x}_\mr{p}^\mr{T}]^\mr{T}$ we design the linear-quadratic pursuit-evasion game as
\begin{equation}\label{eq:lqgame}
\begin{split}
\min_{\mathbf{u}_\mathrm{e}}\max_{\mathbf{u}_\mathrm{p}} ~J_\mr{LQ}(\mathbf{u}_\mathrm{e}, \mathbf{u}_\mathrm{p})=\int_{t_0}^{t_f} L_{\mr{LQ}}(\Delta\mathbf{x}(t),\mathbf{u}_\mathrm{e}(t), \mathbf{u}_\mathrm{p}(t))\mathrm{d}t + \phi_{\mr{LQ}}(\Delta\mathbf{x}(t_f))
\end{split}
\end{equation}
where
\begin{align*}
L_{\mr{LQ}} &:= \norm{\Delta{\mathbf{x}}_\mathrm{e}(t)}_{Q_\mathrm{e}(t)}^2 + \norm{\mathbf{u}_\textrm{e}(t)}_{R_\mathbf{e}(t)}^2 - \norm{\Delta{\mathbf{x}}_\mathrm{p}(t)}_{Q_\mathrm{p}(t)}^2 - \norm{\mathbf{u}_\textrm{p}(t)}_{R_\mr{p}(t)}^2  - \norm{\mathbf{p}_\textrm{e}(t)-\mathbf{p}_\textrm{p}(t)}_{M(t)}^2, \\ \phi_{\mr{LQ}}&:=\norm{\Delta{\mathbf{x}}_\mathrm{e}(t_f)}_{F_\mathrm{e}}^2 - \norm{\Delta{\mathbf{x}}_\mathrm{p}(t_f)}_{F_\mathrm{p}}^2 - \norm{\mathbf{p}_\textrm{e}(t_f)-\mathbf{p}_\textrm{p}(t_f)}_{M_f}^2,
\end{align*}
subject to
\begin{align*}
\Delta\dot{\mathbf{x}}_\mathrm{e}(t)& =A(\mathbf{x}_\mathrm{d}(t))\Delta\mathbf{x}_\mr{e}(t)+B_\mr{e}\mathbf{u}_\mathrm{e}(t), \quad \Delta{\mathbf{x}}_\mathrm{e}(t_0)={\mathbf{x}}_\mathrm{e0}-{\mathbf{x}}_\mathrm{d}(t_0),\\
\Delta\dot{\mathbf{x}}_\mathrm{p}(t)&=
A(\mathbf{x}_\mathrm{d}(t))\Delta\mathbf{x}_\textrm{p}(t)+B_\textrm{p}\mathbf{u}_\mathrm{p}(t), \quad \Delta{\mathbf{x}}_\mathrm{p}(t_0)={\mathbf{x}}_\mathrm{p0}-{\mathbf{x}}_\mathrm{d}(t_0),\
\end{align*}
where $Q_\textrm{e},~R_\textrm{e}, ~ F_\textrm{e},~Q_\textrm{p},~R_\textrm{p}, ~ F_\textrm{p}, ~M,~M_f\succ0$ are weighting matrices. 

In the cost of the game \eqref{eq:lqgame}, the terms $\|\Delta{\mathbf{x}}_\mathrm{e}(t)\|_{Q_\mathrm{e}(t)}^2$, $\|\Delta{\mathbf{x}}_\mathrm{e}(t_f)\|_{F_\mathrm{e}}^2$
incentivize the evader to remain close to the nominal orbit $\mf{x}_\mr{d}$,  the terms $- \|\mathbf{p}_\textrm{e}(t)-\mathbf{p}_\textrm{p}(t)\|_{M(t)}^2$, $- \|\mathbf{p}_\textrm{e}(t_f)-\mathbf{p}_\textrm{p}(t_f)\|_{M_f}^2$  incentivize avoiding the pursuer, and $\norm{\mathbf{u}_\textrm{e}(t)}_{R_\mathbf{e}(t)}^2$ captures the requirement that fuel consumption is minimal. The rest of the terms indicate reciprocal requirements for the pursuer. Moreover, the assumption that the pursuer maximizes \eqref{eq:lqgame} offers a security guarantee, in the sense that if $(\mf{u}_\mr{e}^\star,~\mf{u}_\mr{p}^\star)$ are a saddle-point solution to \eqref{eq:lqgame} then
\begin{equation*}
J_\mr{LQ}(\mathbf{u}_\mathrm{e}^\star, \mathbf{u}_\mathrm{p})\le J_\mr{LQ}(\mathbf{u}_\mathrm{e}^\star, \mathbf{u}_\mathrm{p}^\star)\le J_\mr{LQ}(\mathbf{u}_\mathrm{e}, \mathbf{u}_\mathrm{p}^\star), \quad \forall \mathbf{u}_\mathrm{e},~\mathbf{u}_\mathrm{p},
\end{equation*}
i.e., the cost of the evader is upper bounded by $J_\mr{LQ}(\mathbf{u}_\mathrm{e}^\star, \mathbf{u}_\mathrm{p}^\star)$ irrespective of the pursuer's strategy.

\subsection{Linear-Quadratic Game Solution}

In what follows, we cast \eqref{eq:lqgame} in the nominal linear-quadratic form and obtain its solution. Specifically, the following result characterizes the optimal strategies of the pursuer and the evader through a differential Riccati equation. 

\begin{theorem}\label{th:LQ}
Suppose that the differential Riccati equation
\begin{equation}\label{eq:DRE}
-\dot{S}(t)=\mf{A}^\mr{T}(t)S(t)+S(t)\mf{A}(t)+\mf{Q}(t)-S(t)\mf{B}_\mr{e}R_\mr{e}^{-1}(t)\mf{B}_\mr{e}^\mr{T}S(t)+S(t)\mf{B}_{\mr{p}}R_\mr{p}^{-1}(t)\mf{B}^\mr{T}_\mr{p}S(t),~S(t_f)=\mf{Q}_f,
\end{equation}
has a unique, symmetric, bounded solution $S:[t_0,~t_f]\rightarrow\mathbb{R}^{12\times12}$, where
\begin{equation*}
\mf{A}(t)=\begin{bmatrix*}A(\mathbf{x}_\mr{d}(t)) & 0_{6} \\  0_{6} & A(\mathbf{x}_\mr{d}(t)) \end{bmatrix*},~\mf{B}_\mr{e}=\begin{bmatrix}0_{3} & \frac{1}{m_\mr{e}}I_3 & 0_{3} & 0_{3} \end{bmatrix}^\mr{T}, ~\mf{B}_\mr{p}=\begin{bmatrix}0_{3} & 0_{3} & 0_{3} & \frac{1}{m_\mr{p}}I_3 \end{bmatrix}^\mr{T},
\end{equation*}
and
\begin{equation*}
~ \mf{Q}(t)=\begin{bmatrix}Q_\mr{e}(t)-\mf{M}(t) & \mf{M}(t) \\  \mf{M}(t) & -Q_\mr{p}(t)-\mf{M}(t) \end{bmatrix},~\mf{Q}_f=\begin{bmatrix}F_\mr{e}-\mf{M}_f & \mf{M}_f \\  \mf{M}_f & -F_\mr{p}-\mf{M}_f \end{bmatrix},~\mf{M}(t)=\begin{bmatrix}M(t) & 0_{3} \\ 0_{3} & 0_{3} \end{bmatrix}, ~\mf{M}_f=\begin{bmatrix}M_f & 0_{3} \\ 0_{3} & 0_{3} \end{bmatrix}.
\end{equation*}
Then, under closed-loop information pattern, the game \eqref{eq:lqgame} admits a saddle-point solution $(\mf{u}_\mr{e}^\star,~\mf{u}_\mr{p}^\star)$ given by
\begin{equation}\label{eq:saddle_lq}
\begin{split}
\mf{u}_\mr{e}^\star(\Delta\mf{x}(t), t)&=-R_\mr{e}^{-1}(t)\mf{B}_\mr{e}^\mr{T}S(t)\Delta\mathbf{x}(t), \\
\mf{u}_\mr{p}^\star(\Delta\mf{x}(t), t)&=R_\mr{p}^{-1}(t)\mf{B}_\mr{p}^\mr{T}S(t)\Delta\mathbf{x}(t),
\end{split}
\end{equation}
where $\Delta\mathbf{x}=[\Delta\mathbf{x}_\mr{e}^\mr{T} ~ \Delta\mathbf{x}_\mr{p}^\mr{T}]^\mr{T}$.
\end{theorem}
\begin{proof}
Note that we can write the dynamics of the game \eqref{eq:lqgame} in the compact form
\begin{equation}\nonumber
\Delta\dot{\mathbf{x}}(t)=\mf{A}(t)\Delta{\mathbf{x}}(t)+\mf{B}_\mr{e}\mf{u}_\mr{e}(t)+\mf{B}_\mr{p}\mf{u}_\mr{p}(t).
\end{equation}
In addition, we have $\mf{p}_\mathrm{e}-\mf{p}_\mathrm{p}=[I_3 ~ 0_{3} ](\Delta\mathbf{x}_\mathrm{e}-\Delta\mathbf{x}_\mathrm{p})$, and hence 
\begin{align*}
-(\mf{p}_{\mr{e}}(t)-\mf{p}_{\mr{p}}(t))^\mathrm{T}M(t)(\mf{p}_{\mr{e}}(t)-\mf{p}_{\mr{p}}(t))&=\Delta\mathbf{x}(t)^\mr{T}\begin{bmatrix}-\mf{M}(t) & \mf{M}(t) \\ \mf{M}(t) & -\mf{M}(t)\end{bmatrix}\Delta\mathbf{x}(t), \\ -(\mf{p}_{\mr{e}}(t_f)-\mf{p}_{\mr{p}}(t_f))^\mathrm{T}M_f(\mf{p}_{\mr{e}}(t_f)-\mf{p}_{\mr{p}}(t_f))&=\Delta\mathbf{x}^\mr{T}(t_f)\begin{bmatrix}-\mf{M}_f & \mf{M}_f \\ \mf{M}_f & -\mf{M}_f\end{bmatrix}\Delta\mathbf{x}(t_f).
\end{align*}
Combining these details, we can write \eqref{eq:lqgame} in the compact form:
\begin{equation}\label{eq:cost_compact}
\begin{split}
\min_{\mathbf{u}_\mathrm{e}}\max_{\mathbf{u}_\mathrm{p}} J_\mr{LQ}(\mathbf{u}_\mathrm{e}, \mathbf{u}_\mathrm{p})&=\int_{t_0}^{t_f} \left(\norm{\Delta\mf{x}(t)}_{\mf{Q}(t)}^2 +\norm{\mf{u}_{\mr{e}}(t)}_{R_\mr{e}(t)}^2 - \norm{\mf{u}_{\mr{p}}(t)}_{R_\mr{p}(t)}^2\right)\mathrm{d}t+ \norm{\Delta\mf{x}(t_f)}_{\mf{Q}_f}^2,\\
\textrm{s.t.} \qquad\Delta\dot{\mathbf{x}}(t)&=\mf{A}(t)\Delta{\mathbf{x}}(t)+\mf{B}_\mr{e}\mf{u}_\mr{e}(t)+\mf{B}_\mr{p}\mf{u}_\mr{p}(t).
\end{split}
\end{equation}
Finally, note that \eqref{eq:cost_compact} is a game in the standard linear-quadratic form. Hence, the final result follows from Theorem 6.17 in \cite{bacsar1998dynamic}. \frQED
\end{proof}

\end{document}